\documentclass[pra,amsmath,amssymb, notitlepage, twocolumn,superscriptaddress,longbibliography,nofootinbib,floatfix]{revtex4-1}
\usepackage[pdftex,colorlinks=true,linkcolor=darkblue,citecolor=blue,urlcolor=darkred]{hyperref}
\usepackage{braket}
\usepackage{amsmath,amsfonts, amssymb, amsthm, dsfont}
\usepackage{yfonts}
\usepackage{bm,breqn}
\usepackage{mathrsfs}
\usepackage{array,makecell}
\usepackage{svg}
\usepackage{graphicx}
\usepackage[caption=false]{subfig}
\usepackage{verbatim,soul}

\usepackage{tikz}
\usetikzlibrary{calc}
\usepackage{multirow}
\usepackage{color}
\usepackage[capitalise]{cleveref}

\theoremstyle{definition}
\newtheorem{proposition}{Proposition}
\newtheorem{lemma}{Lemma}
\newtheorem*{lemma*}{Lemma}

\definecolor{darkblue}{rgb}{0.,0.,0.4}
\definecolor{darkred}{rgb}{0.5,0.,0.}

\graphicspath{{figs/}{}}

\newcommand{\refeq}[1]{Eq.~(\ref{#1})}
\newcommand{\reffig}[1]{Fig.~\ref{#1}}
\newcommand{\refsec}[1]{Sec.~\ref{#1}}
\newcommand{\reftab}[1]{Tab.~\ref{#1}}
\newcommand{\refcite}[1]{Ref.~\cite{#1}}

\newcommand{\Tr}{\text{Tr}}

\begin{document}
\title{The Penrose Tiling is a Quantum Error-Correcting Code}
\author{Zhi Li}\affiliation{\PI}
\author{Latham Boyle}\affiliation{\PI}\affiliation{\Edinburgh}
\newcommand*{\PI}{Perimeter Institute for Theoretical Physics, Waterloo, Ontario N2L 2Y5, Canada}
\newcommand*{\Edinburgh}{Higgs Centre for Theoretical Physics, James Clerk Maxwell Building, Edinburgh EH9 3FD, UK} 

\begin{abstract}
The Penrose tiling (PT) is an intrinsically non-periodic way of tiling the plane, with many remarkable properties.  A quantum error-correcting code (QECC) is a clever way of protecting quantum information from noise, by encoding the information with a sophisticated type of redundancy.  
Although PTs and QECCs might seem completely unrelated, in this paper we point out that PTs give rise to (or, in a sense, {\it are}) a remarkable new type of QECC.  In this code, quantum information is encoded through quantum geometry, and any local errors or erasures in any finite region, no matter how large, may be diagnosed and corrected.
We also construct variants of this code (based on the Ammann-Beenker and Fibonacci tilings) that can live on finite spatial tori, in discrete spin systems, or in an arbitrary number of spatial dimensions.  We discuss connections to quantum computing, condensed matter physics, and quantum gravity.
\end{abstract}
\maketitle

\section{Introduction}

Penrose tilings (PTs) \cite{penrose1974role} are a class of tessellations of the two-dimensional (2D) plane, whose beautiful and unexpected properties have fascinated physicists, mathematicians, and geometry lovers of all sorts, ever since their discovery in the 1970s \cite{gardner1977extraordinary, baake2013aperiodic}.  These tilings are intrinsically non-periodic, yet perfectly long-range ordered; and among many other remarkable characteristics, they exhibit a kind of self-similarity and a ten-fold symmetry forbidden in any periodic pattern (see \reffig{fig-QCQECC}(a)).  In the 1980s, they also turned out to be the blueprints for a new class of materials (quasicrystals) discovered in the lab \cite{PhysRevLett.53.1951}, and later found to also occur naturally \cite{bindi2009natural} (forming {\it e.g.}\ in the birth of the Solar System \cite{hollister2014impact}, in lightning strikes \cite{bindi2023electrical}, and in the first atomic bomb test \cite{bindi2021accidental}).

A quantum error-correcting code (QECC) \cite{shor1995scheme} is a way of encoding quantum information with a sophisticated type of redundancy, so that certain errors can be detected and corrected, and the original quantum information can be reconstructed from the disrupted states. 
Such codes play a deep and increasingly wide-ranging role in physics: in quantum computing, where they protect the delicate quantum state of the quantum computer \cite{kitaev2002classical,nielsen2002quantum}; 
in condensed matter physics, where they underpin the notion of topologically-ordered phases (whose ground states form the code space of such a code) \cite{kitaev2003fault,bravyi2010topological,zeng2019quantum}; and even in quantum gravity, where the holographic or gauge/gravity duality \cite{maldacena1999large} may be understood as such a code \cite{Almheiri:2014lwa, Pastawski:2015qua}.

\begin{figure}
  \begin{center}
  \subfloat[]{\includegraphics[width=1.7in]{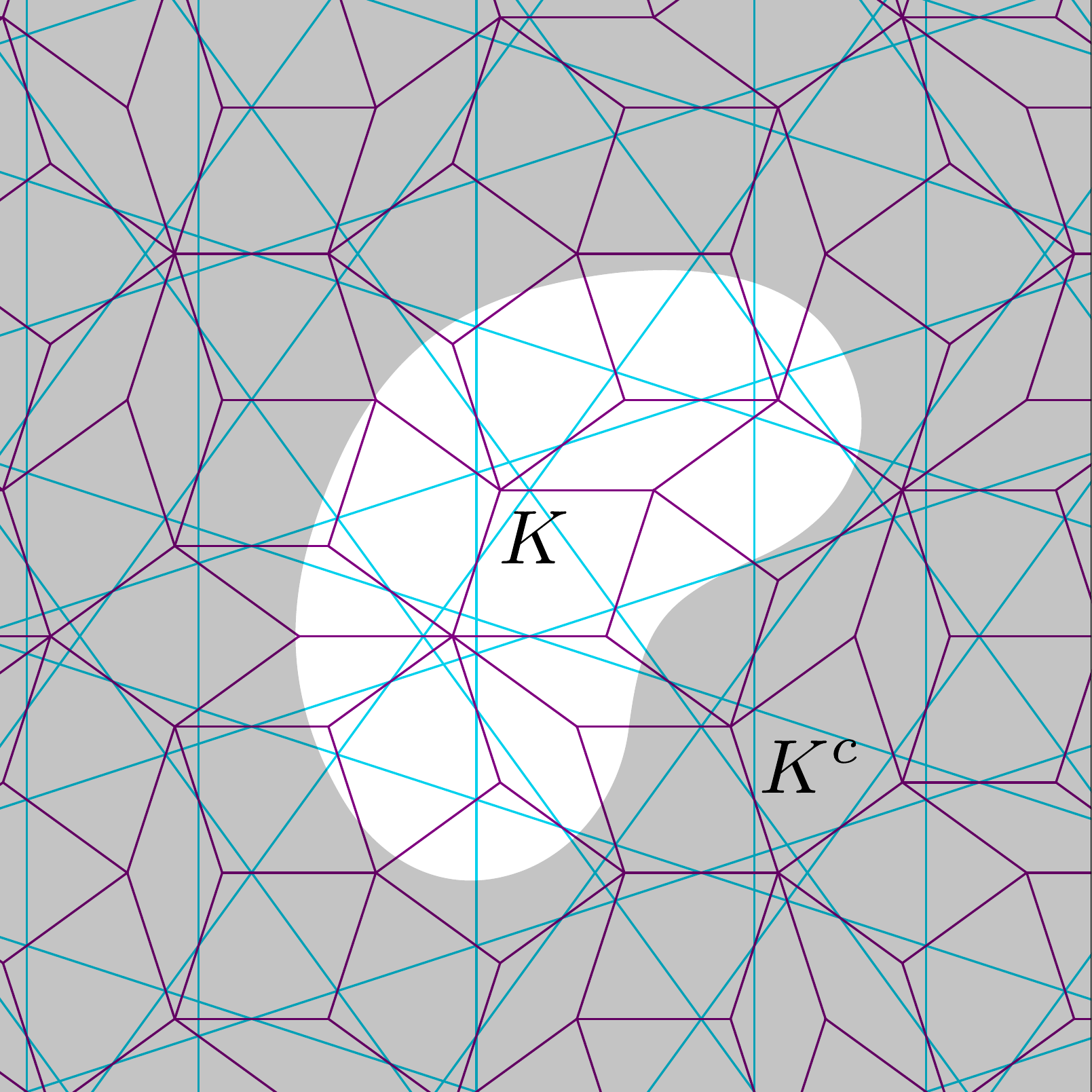}}
    \subfloat[]{\includegraphics[width=1.7in]{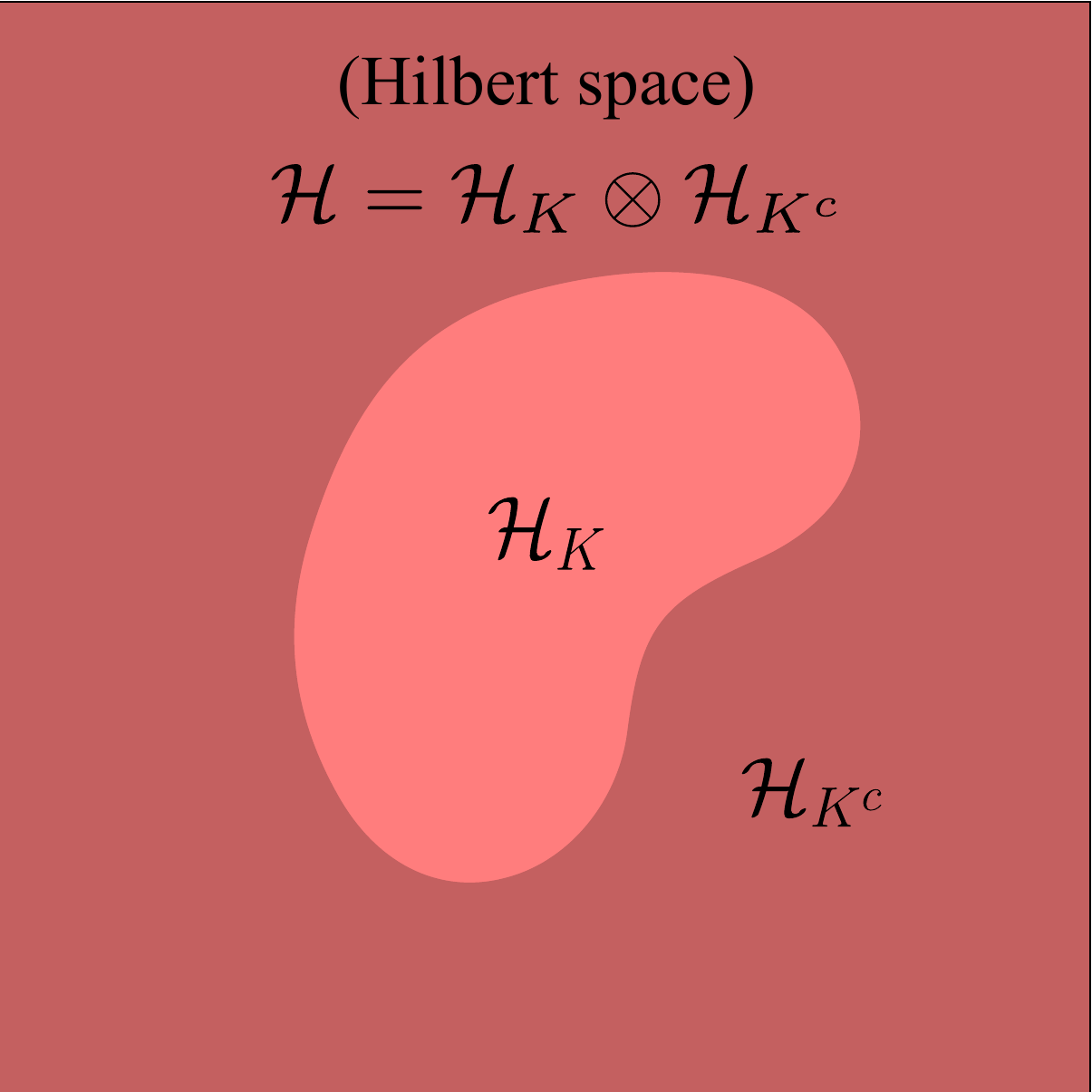}}\\[-1ex]
  \end{center}
  \caption{Parallel between PTs and QECCs.  (a) In a PT: examining a finite region $K$ tells you nothing about which PT you are in, but examining the complementary region $K^c$ allows you to reconstruct full PT. Here, purple lines are the edges of tiles and blue lines are Ammann lines. (b) In a QECC: examining a logical state (code state) in a finite region $K$ tells you nothing about which logical state you're in, but examining the complementary region $K^c$ allows you to reconstruct the logical state on the full space.}
  \label{fig-QCQECC}
\end{figure}

Although PTs and QECCs might seem completely unrelated, we will see there is a deep connection.

On the QECC side, one of the most fundamental insights underlying quantum error correction is the following equivalence between \textit{recoverability} (or \textit{correctability}) and \textit{indistinguishability} \cite{kitaev2002classical, nielsen2002quantum}: arbitrary errors in and erasures of a certain spatial region $K$ are correctable if and only if the region contains no logical information; more precisely, if and only if the various states in the code space are indistinguishable in $K$ (in the sense that their reduced density matrices in $K$ are identical). Colloquially speaking, the quantum information is encoded in a \emph{global} way, rather than in any small region.

On the PT side, there are also analogous concepts of recoverability and indistinguishability that play a key role.  To understand indistinguishability, note that there are actually an infinite number of distinct PTs: they are {\it globally} inequivalent (in the sense that no translations and rotations of one PT can bring it into perfect global agreement with another distinct PT), but they are all {\it locally} indistinguishable, in the sense that any finite patch of any PT, no matter how large, must also appear in any other distinct PT (so no matter how large a finite region one explores, one cannot determine {\it which} PT one is exploring).  As for recoverability: if we erase any finite region $K$ of a PT, no matter how large, the missing region can be uniquely recovered from knowledge of the rest of the tiling (in the complementary region $K^c$).

The notions of indistinguishability and recoverability in PTs indeed smell similar to those in QECCs.  There are, however, crucial differences. The former indistinguishability is a \textit{classical} property concerning geometric configurations and relates \textit{different} spatial regions; while the latter is a property valid for encoded \textit{quantum} states, including their quantum superpositions, and relates to different quantum states in the \textit{same} spatial region.

Nevertheless, in this paper, we show that, because of these indistinguishability and recoverability properties, PTs give rise to (or, in a sense, {\it are}) a remarkable new type of QECC in which the quantum information is encoded through quantum geometry, and any local errors or erasures in any finite region, no matter how large, may be diagnosed and corrected.  
We then show how to use relatives of the Penrose tiling (called the Ammann-Beenker tiling \cite{ammann1992aperiodic, beenker1982algebraic, baake2013aperiodic} and the Fibonacci tiling \cite{grunbaum1987tilings, senechal1996quasicrystals}) to construct variants of this PT QECC that can live on finite spatial tori, in discrete spin systems, or in arbitrary spatial dimension.

\section{Penrose Tilings} \label{sec:qc}
We begin with a brief introduction to Penrose tilings. The properties to be discussed here are quite general and also apply to the Fibonacci and Ammann-Beenker tilings that we will discuss later.

\subsection{Defining the Tilings}

\begin{figure}
\centering
\subfloat[]{\includegraphics[width=0.75\columnwidth]{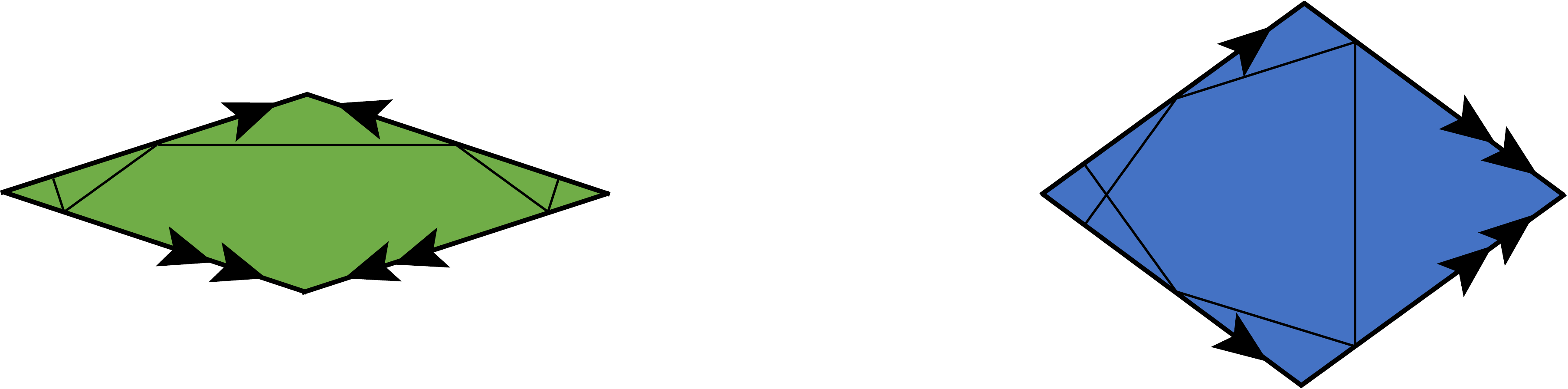}}\\[-1ex]
\subfloat[]{\includegraphics[width=0.9\columnwidth]{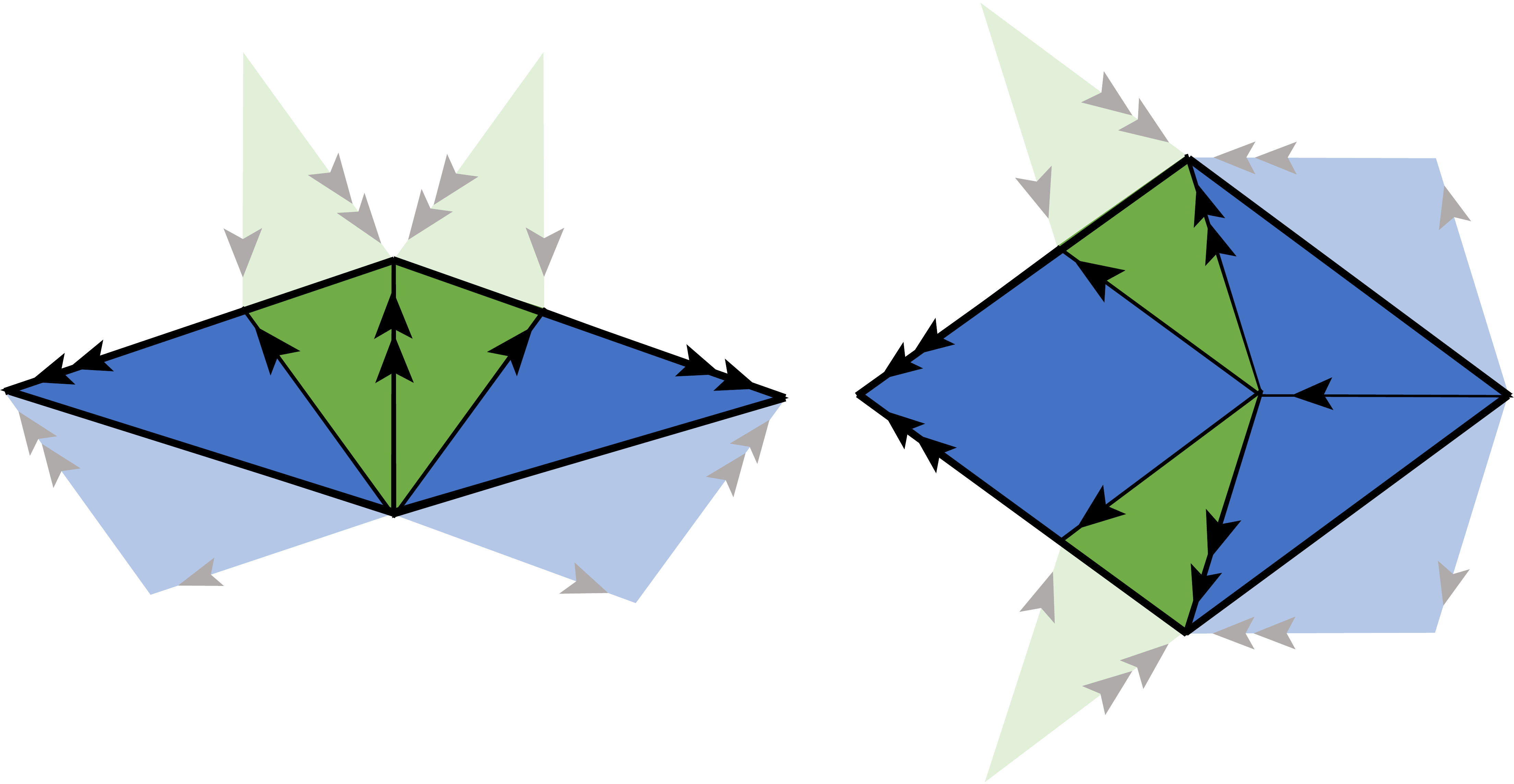}}\\[-1ex]
\caption{(a) Two fundamental rhombi for Penrose tilings, thin (green) and thick (blue), decorated by arrows or Ammann bars to define the matching rule. 
(b) The inflation rule for Penrose tilings. Note that inflation of a parent rhombus produces several half offspring rhombi, which are to be combined with other half offspring rhombi (indicated by lighter color), coming from inflation of neighboring parent rhombi.
}
  \label{fig-PT}
\end{figure}

The Penrose tiling, shown by the purple lines in \reffig{fig-QCQECC}(a), is a tessellation of $\mathbb R^2$ using two types of rhombi: thick (with angle $2\pi/5$) and thin (with angle $\pi/5$).  The edges of the rhombi are decorated by arrows as in \reffig{fig-PT}(a), and the arrows are required to match on all edges in the tiling.  This \emph{matching rule} \cite{levitov1988local,socolar1990locality} forbids {\it e.g.}\ the trivial periodic tiling constructed from just one of the two rhombi.  It is a striking discovery of \refcite{penrose1974role} that such a tessellation is possible, but only if it is aperiodic.

Instead of arrows on the edges, we can equivalently specify the matching rules by decorating the two tiles with a certain pattern of line segments (called ``Ammann bars"), shown in \reffig{fig-PT}(a), and demanding that these Ammann bars form unbroken straight lines as they cross from one tile to the next: see the thin blue lines in \reffig{fig-QCQECC}(a).  The resulting blue pattern of Ammann lines is ``dual" to the purple Penrose tiling: they uniquely determine each other (see appendix \ref{app-PTA}).

\subsection{Inflation and Deflation}\label{sec-infdef}

There is another useful description of Penrose tilings. We define an \emph{inflation} rule, where we cut each green or blue rhombus into a certain fixed pattern of smaller green and blue rhombi, as shown in Fig.\ref{fig-PT}(b).  
By cutting up all the tiles of any legal PT (the ``parent" PT) in this way and rescaling it by the golden ratio $\frac{1}{2}(\sqrt{5}+1)$, we obtain another legal PT (the ``offspring" PT). This inflation process may be iterated indefinitely many times.

We can also consider the inverse process, \emph{deflation}, in which we glue together the tiles of an offspring PT into ``supertiles," again according to the pattern shown in Fig.\ref{fig-PT}(b), to recover the parent PT from which it descended.  In a legal PT, there is a \emph{unique} way to group all the tiles into supertiles in this way, so one can unambiguously determine the parent PT from its offspring, and this deflation process may be iterated indefinitely.

\subsection{Indistinguishability and Recoverability}\label{sec-qc-indis}

Importantly, the (arrow or Ammann line) matching rule does \emph{not} fully determine the PT (even up to translation and rotation): there are actually uncountably many globally inequivalent PTs!

However, these inequivalent PTs are closely related to each other. 
Mathematically, they belong to the same \textit{local indistinguishability class}: any finite patch present in one Penrose tiling $T$ must also appear in any other Penrose tiling $T'$ \cite{baake2013aperiodic}\footnote{Note that the term ``locally indistinguishable" used here is synonymous with the earlier term ``locally isomorphic" introduced in \cite{levine1986quasicrystals}.  Both terms are now in common use.}.  Therefore, in the absence of an absolute Euclidean reference frame (specifying the position of the origin, and the orientation of the $\hat{x}$ and $\hat{y}$ directions), inequivalent PTs only differ in their ``global" behavior; one cannot distinguish them by inspecting a finite region, no matter how large it is.  In this paper, we will make use of a stronger, quantitative version of local indistinguishability: not only do all finite patches in $T$ appear in $T'$ and vice versa, but the relative frequencies of different finite patches are also the same. In fact, the relative frequencies can be calculated solely from the inflation rule, see appendix \ref{app-indis} for details.

Another important feature is \emph{local recoverability}: the pattern in any finite region $K$ can be uniquely recovered from the pattern in the complementary region $K^c$.  This is because one can extend the Ammann lines from $K^c$ into $K$, thus recovering the Ammann lines on the whole plane, and hence recovering the whole PT.

Note that these are properties of the geometric configurations, and should not be confused with the local indistinguishability and recovery in QECCs, which are properties of the quantum states. 
Nevertheless, as the key point of this paper, we will leverage these classical properties to construct quantum states that are indeed quantum error-correcting.

\section{The Penrose Tiling as a QECC}\label{sec:construction}

In this section, we show how the PT yields a QECC.  We begin with a brief, informal introduction to QECCs.  (For more, see appendix \ref{app-qecc}.)

\subsection{QECCs: A  Brief Introduction}

Suppose the quantum information we want to protect are quantum states in the Hilbert space ${\cal H}_{0}$.  A QECC works by storing this state with a carefully chosen type of redundancy so that certain errors can be identified and corrected.  More precisely, the Hilbert space ${\cal H}_{0}$ of ``bare" or ``logical" quantum states is ``encoded" by embedding it in an enlarged Hilbert space ${\cal H}$ as a carefully chosen subspace $\mathcal C$, called the code space.

In this paper, the errors that can be corrected will be the erasure of any arbitrary finite spatial region $K$. As a result, arbitrary errors in $K$ will also be correctable.  It is a fundamental fact \cite{kitaev2002classical, nielsen2002quantum} that the erasure of region $K$ is correctable if and only if $K$ contains no logical information, {\it i.e.}, if and only if the various states in the code space ${\cal C}$ are \emph{indistinguishable} in $K$: see \refeq{eq-QEC-def-main} below.
Moreover, note that the spatial region $K$ is not fixed: we could decompose the whole space into the union of many such $K$'s, each satisfying the QECC condition and containing no information. Hence, in a QECC, the quantum information is stored in a ``global" way.

Thinking of ${\cal H}$ as ${\cal H}_{K}\otimes{\cal H}_{K^c}$ (where ${\cal H}_{K}$ and ${\cal H}_{K^c}$ are the Hilbert spaces for $K$ and the complementary region $K^{c}$), indistinguishability says that:
\begin{equation}\label{eq-QEC-def-main}    \Tr_{K^c}\ket{\xi}\bra{\xi}=\Tr_{K^c}\ket{\xi'}\bra{\xi'}
\end{equation}
$\forall\ket\xi,\ket{\xi'}\in\mathcal C$ that are normalized (where $\Tr_{K^c}$ means tracing over ${\cal H}_{K^c}$).
If the space $\mathcal{C}$ is spanned by states $\ket{\psi_i}$, then \refeq{eq-QEC-def-main} is equivalent to 
\begin{equation}\label{eq-QECcriteria}
\Tr_{K^c}\ket{\psi_i}\bra{\psi_j}=\braket{\psi_j|\psi_i}\rho_K,~~\forall i,j.
\end{equation}
Here, it is crucial that $\rho_K$ (an operator in ${\cal H}_{K}$) is independent of $i$ and $j$.  In the above criteria, $\{\ket{\psi_i}\}$ could be unnormalized, non-orthogonal or over-complete. For a proof, see appendix \ref{app-qecc}.
We emphasize that the states being protected are not limited to the ``basis" states $\ket{\psi_i}$, but can also be arbitrary quantum superpositions of such states, namely, any states in $\mathcal C$.

With this background, we are now ready to construct the PT QECC, capable of correcting arbitrary erasures and errors in any finite spatial region $K$.

\subsection{Constructing the PT QECC}

\begin{figure*}
    \centering
    \includegraphics[width=0.8\linewidth]{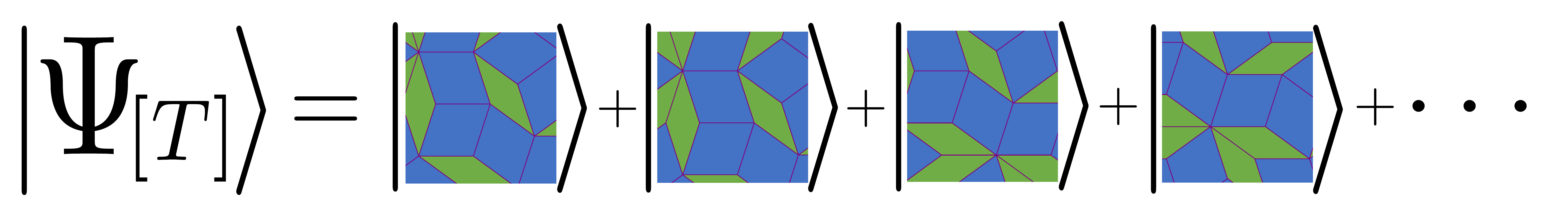}
    \caption{An illustration for the wavefunction \refeq{eq:codegs}. 
    Given an (infinite) Penrose tiling $T$, each term on the right-hand side represents a Euclidean transformed version of the original tiling, denoted by $gT$. 
    Here, four patches are drawn from the same tiling $T$, serving to illustrate the relative translations and rotations among the $gT$'s.
    }
    \label{fig-wavefunction}
\end{figure*}

We consider the set of Penrose tilings $\mathcal{T}$.  If $T$ denotes a particular PT in $\mathcal{T}$, then $gT$ denotes the PT obtained from $T$ by the 2D Euclidean transformation (translation and/or rotation) $g$, and $[T]=\{gT\}$ denotes the equivalence class of all PTs that are equivalent to $T$ up to 2D Euclidean transformations.

We can regard a tiling $T$ as a state $\ket{T}$ in a quantum mechanical Hilbert space ${\cal H}$.  Dividing the plane $\mathbb{R}^{2}$ into any spatial region $K$ and the complementary region $K^c$ divides the tiling $T$ into the corresponding portions $T_{K}$ and $T_{K^c}$ which lie in $K$ and $K^c$, respectively.  This induces a decomposition of the Hilbert space ${\cal H}={\cal H}_{K}\otimes{\cal H}_{K^{c}}$, and a corresponding decomposition of the state $\ket{T}=\ket{T}_{K}\ket{T}_{K^c}$, where $\ket{T}_{K}\in{\cal H}_{K}$ and $\ket{T}_{K^c}\in{\cal H}_{K^c}$ only depend on $T_{K}$ and $T_{K^c}$, respectively.  If two tilings $T$ and $T'$ are distinct (and here we mean distinct in the {\it presence} of an absolute reference frame, so that even two tilings that merely differ by an overall Euclidean transformation are distinct in general) they are represented by orthogonal states in ${\cal H}$: 
\begin{equation}
  \label{delta_T_Tp}
    \braket{T'|T}=\delta(T',T),
\end{equation}
and similarly, if $T$ and $T'$ are distinct in $K$ ({\it i.e.}\ if $T_{K}\neq T_{K}'$), then $\ket{T}_{K}$ and $\ket{T'}_{K}$ are orthogonal in ${\cal H}_{K}$.

For each equivalence class $[T]$ define the wavefunction
\begin{equation}\label{eq:codegs}
    \ket{\Psi_{[T]}}=\int dg \ket{gT},
\end{equation}
where we superpose over all Euclidean transformations $g$, so $\ket{\Psi_{[T]}}$ only depends on $[T]$, see \reffig{fig-wavefunction} for illustration.  The main claim of this paper is that the states $\ket{\Psi_{[T]}}$ form an orthogonal basis for the code space ${\cal C}\subset{\cal H}$ of a QECC that corrects arbitrary errors or erasures in any finite region $K$.

To understand the claim, let us check the criterion \refeq{eq-QECcriteria} for a QECC.  In our case, Eqs.~(\ref{delta_T_Tp}, \ref{eq:codegs}) imply that $\braket{\Psi_{[T]}|\Psi_{[T']}}=0$ when $[T]\neq[T']$, so we need to check:
\begin{equation}\label{eq:QECcheck}
    \Tr_{K^c}\ket{\Psi_{[T]}}\bra{\Psi_{[T']}}=
    \begin{cases}
      0, & \text{~if~}[T]\neq [T']\\
    \braket{\Psi_{[T]}|\Psi_{[T]}}\rho_{K}, & \text{~if~} [T]=[T']\\
    \end{cases}.
\end{equation}

It turns out the first condition in \refeq{eq:QECcheck} follows from the geometric PT recoverability property discussed in \refsec{sec-qc-indis}. Indeed, with \refeq{eq:codegs} in mind, the vanishing of $\Tr_{K^c}\ket{\Psi_{[T]}}\bra{\Psi_{[T']}}$ says two classical configurations $T$ and $T'$ belonging to different classes $[T]\neq [T']$ must also differ on $K^c$. In other words, the configuration in $K^c$ should uniquely determine the equivalence class it belongs to. This is guaranteed by PT recoverability. 

To check the second condition in \refeq{eq:QECcheck}, we perform the explicit calculation:
\begin{equation}\label{eq:TT}
 \begin{aligned}
     \Tr_{K^c}\ket{\Psi_{[T]}}\bra{\Psi_{[T]}}&=\!\iint\!\!dg' dg~\delta(gT,g'T) \ket{gT}_K\!\bra{g'T}_K\\
     &=\!\int\!\! dg''\;\delta(T,g''T)\cdot\!\int\!\! dg \ket{gT}_K
     \!\bra{gT}_K.
\end{aligned}   
\end{equation}
Here, the first equality is again due to PT recoverability: 
in order for $\Tr_{K^c}\ket{gT}\bra{g'T}$ to be nonzero, $gT$ and $g'T$ must match in $K^c$, and hence globally, giving us the delta function. The second equality follows since $\delta(gT,g'T)$ only depends on the difference $g''=g^{-1}g'$.  Note that the factor $\int dg \delta(gT,T)$ is proportional to $\braket{\Psi_{[T]}|\Psi_{[T]}}$ by the same argument, so it is enough to show that $\rho_{K, [T]}=\int dg \ket{gT}_K\bra{gT}_K$ is actually $T$-independent. 

This is precisely the PT local indistinguishability property. Indeed, $\rho_{K, [T]}$ is simply the classical mixture of all local patterns $T_{K}$ that could appear in $T$ by rotating and translating $K$, weighted by their relative frequencies.  Local indistinguishability, that the local patterns $T_{K}$ (including their relative frequencies) are the same for all tilings $T$, implies that $\rho_{K, [T]}$ is $T$-independent.

In summary, by utilizing the properties of Penrose tilings, we have constructed a QECC that is capable of correcting arbitrary errors and erasures in any finite region. 
In this code, each basis state $\ket{\Psi_{[T]}}$ for the codespace ${\cal C}$ is a quantum superposition over all tilings $|T\rangle$ in the equivalence class $[T]$; and a typical codestate is, moreover, a quantum superposition of such basis states ({\it i.e.}\ of distinct equivalence classes $[T], [T'],\ldots$).  Thus the quantum information is encoded in the quantum geometry of how the Penrose tilings are superposed.

\section{Discrete Realization via the Fibonacci Quasicrystals}\label{sec:discrete}
The above construction uses a continuum degree of freedom, but we will now see this is not essential: In this section, we show how to construct similar QECCs on discrete systems (spin chains), using 1D ``Fibonacci Quasicrystals", a class of 1D analogs of 2D PTs.

\subsection{1D Fibonacci Quasicrystals}

\begin{figure}
\centering
\subfloat[]{\includegraphics[width=\columnwidth]{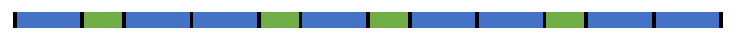}}\\
\subfloat[]{\includegraphics[width=0.7\columnwidth]{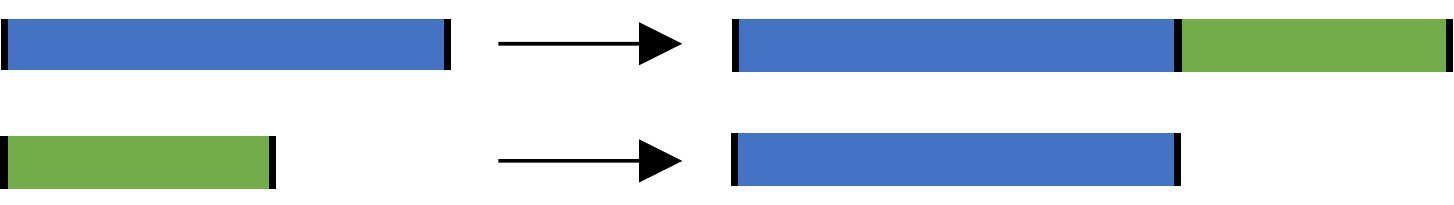}}\\[-1ex]
  \caption{(a) A finite piece of a 1D Fibonacci quasicrystal (scaled to fit), also represented as $LSLLSLSLLSLL$ or 101101011011.
  (b) Inflation rule $(L,S)\to(LS,L)$.
  }
  \label{fig-inflation1D}
\end{figure}

Consider a tessellation of $\mathbb R$ composed of two types of intervals, $L$ (long) and $S$ (short), with relative length equal to the golden ratio: $L=\frac{\sqrt{5}+1}{2}S$. 
For such tilings, we define an inflation rule as:
\begin{equation}
    (L,S)\to(LS,L),
\end{equation}
see \reffig{fig-inflation1D}(b).
Similar to the discussion in \refsec{sec-infdef}, we define a 1D Fibonacci quasicrystal as a tiling using $L$s and $S$s, such that one can perform the corresponding deflation process unambiguously and indefinitely.  The requirements of unambiguous and indefinite deflation strongly restrict the pattern. 
Consequently, there must be no pair of adjacent short tiles $SS$ or triple of adjacent long tiles $LLL$, etc\footnote{The requirement of no adjacent double $S$ resembles the Rydberg blockade \cite{PhysRevLett.85.2208} in the study of quantum spins chains (see also \cite{PhysRevB.106.094202}). However, here we have infinitely many forbidden patterns from different levels of deflations.}, and the ratio of numbers of $L$s and $S$s equals the golden ratio $\frac{\sqrt{5}+1}{2}$ asymptotically
(from which it follows that the tiling must be aperiodic). 

The Penrose tilings and the Fibonacci quasicrystals are closely related: 
the Ammann lines on a PT divide into 5 parallel subsets (each of which is parallel to one of the sides of a regular pentagon),
and if we focus on one of these parallel subsets, we find that the spacings among those lines precisely form a 1D Fibonacci quasicrystal.

\subsection{Discrete QECC from Symbolic Substitutions}

The structure of the Fibonacci quasicrystal also shows up in certain discrete systems. More precisely, we can equivalently represent Fibonacci quasicrystals as two-sided infinite bit strings by replacing each long tile $L$ with the digit 1 and each short tile $S$ with a digit 0. The inflation rule is now:
\begin{equation}
    (1,0)\to(10,1).
\end{equation}
The inflation rules in discrete systems are usually called \emph{symbolic substitutions}.

At the level of bit strings, the crucial properties enabling our QECC construction still hold.
More precisely, there are still infinitely many globally inequivalent ({\it i.e.}\ not related by translation) bit strings $\{F\}$ possessing local indistinguishability and recoverability: 
\begin{itemize}
    \item any finite substring of $F$ also appears in any other string $F'$, and the relative frequencies of different finite substrings also match (see Appendix \ref{app-indis});
    \item any finite substring of $F$ (in the finite region $K$) can be recovered from the remainder of $F$ (in the complementary region $K^c$), except when $F$ is singular\footnote{Among infinitely many strings $F$, only two are singular: the two strings which are entirely reflection symmetric about the origin, except for the central two digits, which are either $LS$ or $SL$.} (see Appendix \ref{sec-recover-1D-infinite}). 
\end{itemize}

Therefore, we can utilize these bit strings to construct discrete many-body wavefunctions analogous to \refeq{eq:codegs}, which form the basis for a similar QECC.  Following the spirit of \refeq{eq:codegs}, we construct the following wavefunction for each equivalence class of bit strings $[F]$:
\begin{equation}\label{eq:1Dwavefunction}
    \ket{\Psi_{[F]}}\propto\sum_{x=-\infty}^{\infty} \ket{x+F},
\end{equation}
where $x+F$ means translation of $F$ by $x$. It can be thought of as a wavefunction for a quantum spin chain living on the integer lattice $\mathbb{Z}$.  Following the arguments in \refsec{sec:construction}, these quantum states $\{\ket{\Psi_{[F]}}\}$ span the code subspace ${\cal C}$ of a QECC correcting arbitrary errors in any arbitrary finite interval.

We compute the entanglement entropy of these code states in Appendix \ref{app-entropy}.

\section{Finite Realization via Ammann–Beenker tilings}\label{sec:AB}

Given that our constructions \refeq{eq:codegs} and \refeq{eq:1Dwavefunction} live on strictly infinite spaces, it is tempting to ask whether such quasicrystal-inspired wavefunctions can be realized on finite systems, e.g., a torus or a ring.  However, there is a fundamental obstruction for PTs: the arrows in \reffig{fig-PT} define an \emph{aperiodic matching rule}: 
any tiling that locally resembles the Penrose tiling, i.e., obeys the matching rule, must be aperiodic, and hence can never be realized on a torus.
Indeed, any PT-like configuration on a torus must contain at least two ``defects" \cite{entin1988penrose}.

\begin{figure}
    \centering
    \subfloat[]{\includegraphics[width=0.38\linewidth]{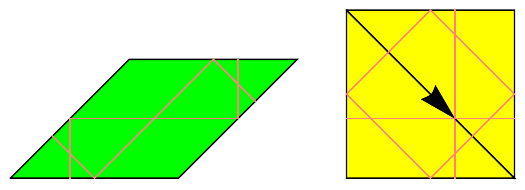}}\\[-1ex]
    \subfloat[]{\includegraphics[width=\linewidth]{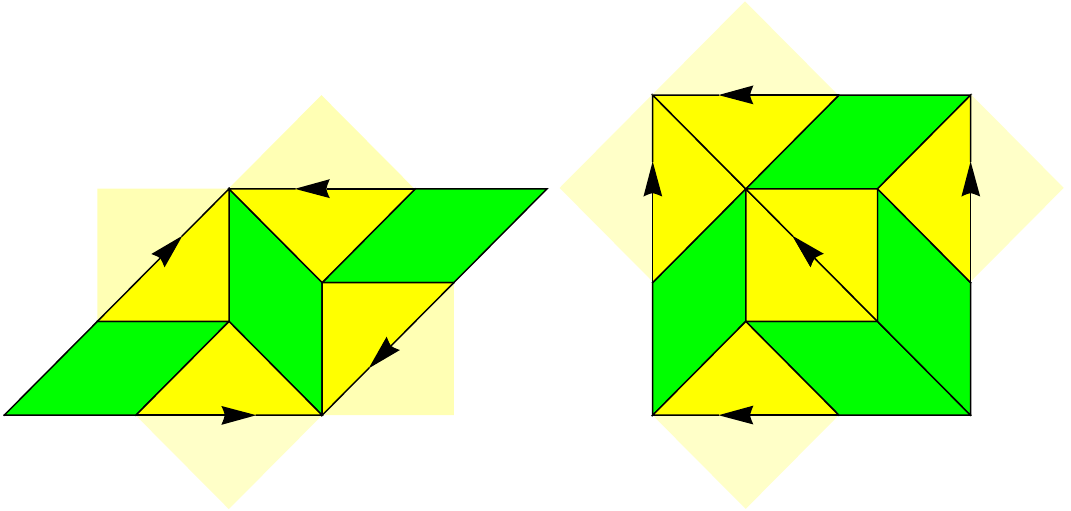}}\\[-1ex]
    \caption{(a) The two fundamental tiles in the Ammann–Beenker tiling, decorated by Ammann bars indicating a matching rule.  The arrow in each square shows its orientation (so it is only symmetric under reflection across the corresponding diagonal). 
    (b) The corresponding inflation rule. }
    \label{fig:ABinflationrule}
\end{figure}

Luckily, there are some tilings where this obstruction can be circumvented.
One example is the Ammann–Beenker (AB) tiling \cite{ammann1992aperiodic, beenker1982algebraic, grunbaum1987tilings, baake2013aperiodic}.
The AB tiling may be defined by the inflation rule shown in \reffig{fig:ABinflationrule}\footnote{Although we draw arrows and Ammann bars on the fundamental tiles, the tiling introduced here is called the {\it undecorated} AB tiling \cite{baake2013aperiodic}. 
The closely related {\it decorated} AB tiling has additional decorations, so that the local matching rule actually forces aperiodicity, which we do not want.}.
AB tilings share many similar properties with PTs. For example, they are aperiodic, are characterized by a kind of quasi-rotational symmetry forbidden in any periodic crystal (eight-fold in the AB case, ten-fold in the PT case),\footnote{We say a tiling $T$ has $n$-fold quasi-rotational symmetry if any finite patch $T_{K}$ in the tiling, no matter how large, also occurs (indeed, an infinite number of times) elsewhere in the tiling, and rotated by every integer multiple of $2\pi/n$.
} and each tile comes with Ammann bars \cite{socolar1989simple, boyle2022coxeter}, which must join into straight, unbroken Ammann lines.

However, unlike the PT, the AB tiling does {\it not} come with aperiodic matching rules \cite{cmp/1104162601}:
given any finite set of allowed local patterns (no matter how large) in the AB tilings, there still exist periodic tilings in $\mathbb R^2$, known as periodic approximants \cite{duneau1989approximants, goldman1993quasicrystals}, whose local patterns are all drawn from this set.
Such periodic, AB-like tilings are \emph{not} genuine AB tilings, which must be aperiodic (by virtue of the inflation/deflation rule that defines them). 
Nevertheless, we can make use of these periodic AB-like tilings (which can be constructed to be locally indistinguishable from genuine AB tilings, up to any arbitrarily chosen scale) to construct quantum error-correcting wavefunctions, in a fashion similar to \refeq{eq:codegs}.

\begin{figure}
    \centering
    \subfloat[]{\includegraphics[width=\linewidth]{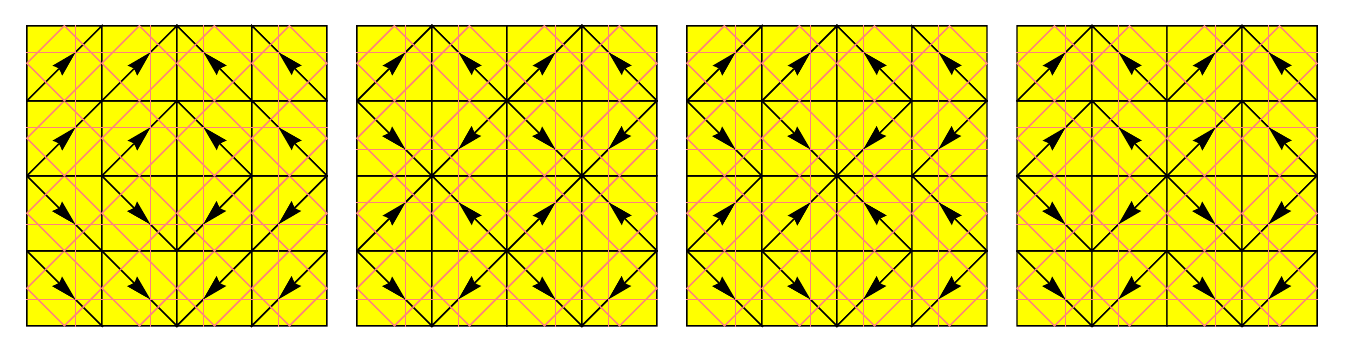}}\\[-1ex]
    \subfloat[]{\includegraphics[width=0.8\linewidth]{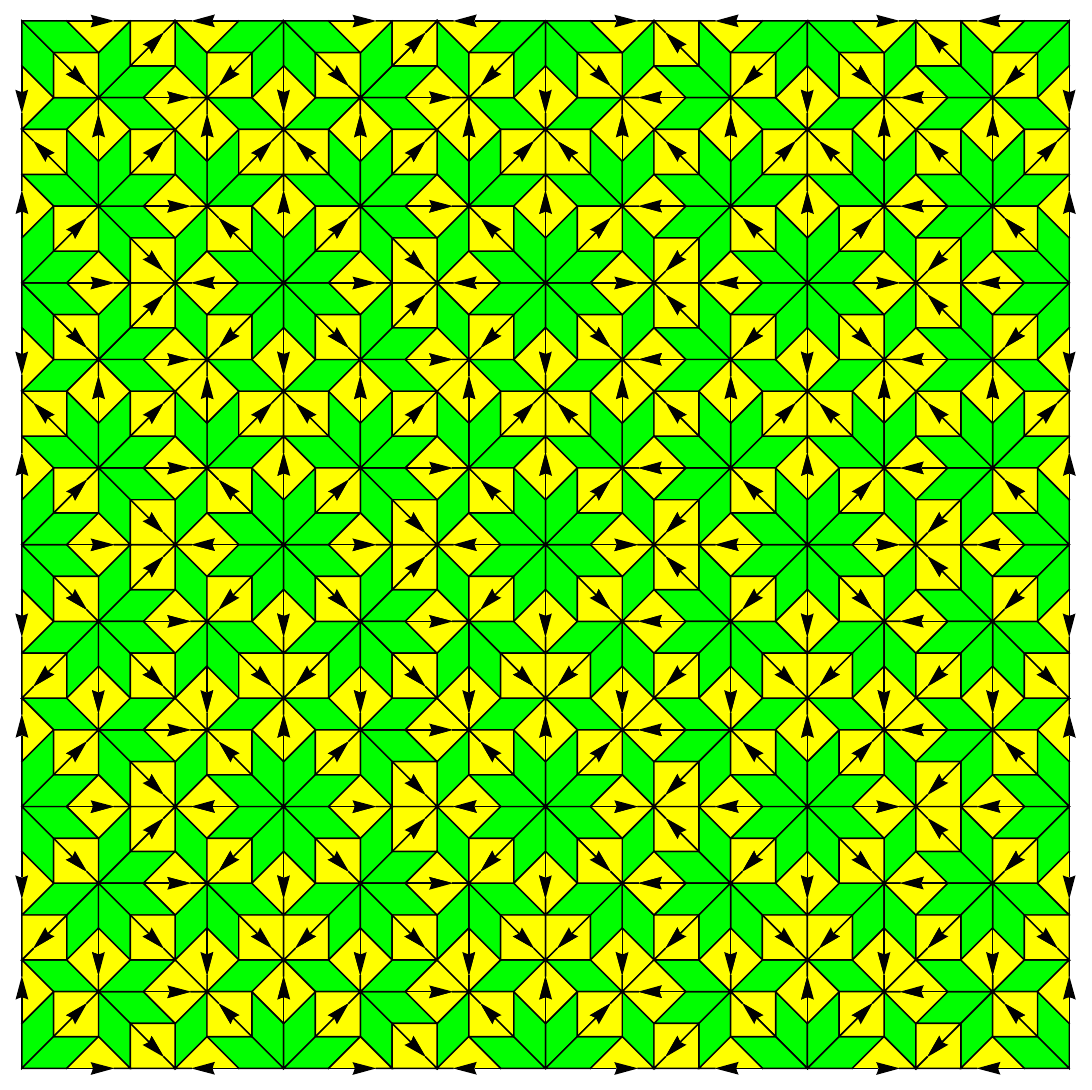}}\\[-1ex]
    \caption{Error-correcting code via Ammann–Beenker tilings. (a) Initial tilings denoted by $A_i^{(0)}$ $(i=1,2,3,4)$ on the torus, together with their Ammann lines. (b) The pattern after inflating $A_1^{(0)}$ twice. It locally resembles a genuine AB tiling.
    }
    \label{fig:AB}
\end{figure}

First, we define four ``initial" tilings $A_i^{(0)}$ $(i=1,2,3,4)$ of a torus, shown in \reffig{fig:AB}(a). 
The tilings are designed so that the Ammann bars form straight unbroken Ammann lines on the torus, and squares of different orientations appear with equal frequency.
Next, we inflate them by repeating the AB inflation rule $n$ times, to obtain four tilings $A_i^{(n)}$, see \reffig{fig:AB}(b) for an illustration of $A_1^{(2)}$.
As a convention, we fix the size of each elementary tile (e.g., each edge has unit length), so the size of $A_i^{(n)}$ grows exponentially with $n$.

Interestingly, $A_i^{(n)}$ for $n\geq 2$ are tilings of the torus that locally resemble genuine AB tilings and are also locally indistinguishable from each other.
More precisely: 
\begin{itemize}
    \item any pattern of $A_i^{(n)}$ inside a disk of radius $r_n=\Theta((1+\sqrt{2})^n)$ ($\propto$ the linear size of $A_i^{(n)}$) also appears in any genuine AB tiling and vice versa.
    \item the number of appearances of a radius $r_n$ pattern in $A_i^{(n)}$ only depends on the pattern itself and is independent of $i$.
\end{itemize}
The proof can be found in Appendix \ref{app-ABfinite}.

Now, we can construct wavefunctions
\begin{equation}
    \ket{\Psi_i^{(n)}}\propto\int dg \ket{gA_i^{(n)}}.
\end{equation}
Here, we only integrate $g$ over the {\it translation} group of the torus (since if we also included rotations, $A_3$ and $A_4$ would become equivalent).

These wavefunctions $\ket{\Psi_i^{(n)}}$ span the code space ${\cal C}$ of a QECC capable of correcting erasures of any radius $r_n$ disk.
As in the PT QECC construction, to verify this, we need only check that the tilings have the needed local indistinguishability and recoverability properties.  The former property is what we just explained; the latter again follows from the Ammann lines, see appendix \ref{app-ABA}.

In appendix \ref{app-finite}, combining the ideas of Secs. \ref{sec:discrete} and \ref{sec:AB}, we construct a QECC that is both finite and discrete and generalizes to any number of spatial dimensions.

\section{Discussion}

In this paper, we leverage the fascinating properties of quasiperiodic tilings (including PTs, AB tilings, and Fibonacci quasicrystals), to construct novel QECCs.  By quantumly superposing PT configurations, we construct a QECC capable of correcting errors in any finite spatial region $K$. Our construction is based on quite general properties of quasiperiodic tilings (namely local indistinguishability and recoverability) and thus remains valid for other such tilings.  Using Fibonacci quasicrystals, we construct a discrete version of this QECC, which can in principle be realized on spin chains.  With the help of AB tilings, we showed how similar QECC codes could be realized in systems of finite spatial extent (on a torus).

It is instructive to draw an analogy with the celebrated toric code \cite{kitaev2003fault}.
Toric code wavefunctions can be written as superpositions of loop configurations on a torus.  Based on the parities of the intersection numbers of these loops with two nontrivial cycles of the torus, the loop configurations are classified into four topologically distinct classes, giving a QECC with a four-dimensional code space.
Our wavefunction \refeq{eq:codegs} shares many properties with toric code wavefunctions:
both wavefunctions are superpositions of geometrical patterns and the logical information is encoded in the global behavior of the patterns;
both have notions of local recoverability and local indistinguishability at the level of geometrical patterns\footnote{For loop configurations, local recoverability means one can determine the parities of the intersection numbers even if a local region of the loop is erased; and local indistinguishability means the pattern of the loop inside a local region says nothing about the global topology (or linking number parities) of the loop.}.
Moreover, both are long-range entangled in the sense that the wavefunction cannot be prepared by finite-depth geometrically-local unitary circuits \cite{PhysRevLett.97.050401}.

On the other hand, while the toric code can correct erasures even if the erased region is noncontiguous and comparable in total to the size of the torus (as long as it does not contain topologically nontrivial loops); 
our finite-size toric construction in \refsec{sec:AB} and the proof therein relies on the fact that the errors are contained in a single contiguous region (whose size again scales linearly with the system size). 
Another difference is that our code space cannot be realized as the space of ground states of any local Hamiltonian.
To see this, note that in \refeq{eq:codegs}, we could insert a configuration-dependent phase factor and define
\begin{equation}
    \ket{\Tilde\Psi_{[T]}}=\int dg e^{i\theta(gT)}\ket{gT}.
\end{equation}
Following a similar calculation, $\ket{\Tilde\Psi_{[T]}}$ has the same reduced density matrix $\rho_{K}$ as $\ket{\Psi_{[T]}}$ for any finite region $K$. 
Therefore, any local Hamiltonian must have the same energy on $\ket{\Psi_{[T]}}$ and $\ket{\Tilde\Psi_{[T]}}$.
One might try to enlarge the code space ${\cal C}$ by including $\{\ket{\Tilde\Psi_{[T]}}\}$ as additional basis vectors, but this does not work, since then superpositions of $\ket{\Psi_{[T]}}$ and $\ket{\Tilde\Psi_{[T]}}$ no longer have the same reduced density matrix $\rho_{K}$.
In this regard, the construction \refeq{eq:codegs} might be better considered as a ground state for the matching rule under a \emph{gauge constraint} imposing translational invariance.  

It will be interesting to revisit the question of how best to create such QECCs in the lab, implement error correction algorithms, and carry out encoding, decoding, and logical operations in this context.

To end on a more speculative note, we mention several hints suggesting that the PT QECC discussed here may capture something about quantum gravity, and the way that the quantum gravitational microstates underlying a spacetime related to (or encode) that spacetime.  (i) First, it has been realized that the holographic picture of quantum gravity in hyperbolic space \cite{maldacena1999large} is itself a kind of quantum error correcting code \cite{Almheiri:2014lwa}; and, moreover, that when one discretizes the hyperbolic space on a tiling that preserves a large discrete subgroup of its original isometry group \cite{Pastawski:2015qua}, this tiling naturally decomposes into a stack of Penrose-like (or Fibonacci-like) aperiodic tilings \cite{Boyle:2018uiv}.  (ii) Second, as emphasized above, the states $\ket{T}$ and $\ket{gT}$ are only distinguishable in the presence of an absolute reference frame.  However, one of the central ideas underlying Einstein's theory of gravity is the principle of general covariance which, physically, asserts that there is {\it no} such absolute reference frame (and, mathematically, asserts that diffeomorphism invariance is an exact gauge symmetry of the laws of nature).  In other words, if one asks {\it where} the ``Penrose tiling microstate" is situated in space then, in order to respect diffeomorphism invariance, the answer must be: it is in a superposition of all possible ways it could be situated -- i.e. it is precisely in one of our code states \refeq{eq:codegs}!  (iii) Third, the intrinsic phase ambiguity about {\it how} to superpose microstates in this way -- {\it i.e.}\ the freedom to define the inequivalent bases $\{\ket{\Psi_{[T]}}\}$ vs $\{\ket{\Tilde{\Psi}_{[T]}}\}$, corresponding to an infinite number of distinct, inequivalent ways to embed the code space ${\cal C}$ in the larger Hilbert space ${\cal H}$ -- seems to reflect the fact that in spacetime there is an intrinsic ambiguity (again due to diffeomorphism invariance) about how to define the ``zero particle state" in that spacetime, leading to an infinite number of inequivalent vacuum states (which is one way to understand {\it e.g.}\ the phenomenon of Hawking radiation from a black hole \cite{birrell1984quantum, mukhanov2007introduction}).  (iv) Fourth, it is natural to wonder whether the analogous QECC built from the four-dimensional Elser-Sloane tiling \cite{elser1987highly} (the beautiful and essentially unique 4D cousin of the 2D Penrose tiling \cite{boyle2022coxeter}, which can be obtained by taking a maximally-symmetric 4D slice of the remarkable 8-dimensional $E_{8}$ root lattice) might provide a particularly interesting model for 4D spacetime.

\begin{acknowledgments} 
We thank Hilary Carteret, Timothy Hsieh, and Beni Yoshida for their helpful discussions.  We also want to thank Hilary Carteret \cite{Carteret} and David Chester \cite{sym14091780} for bringing to our attention their very creative works (which are independent of one another, and orthogonal to the topic and approach in this paper) suggesting possible connections between topological quantum computation and quasicrystals.
Research at Perimeter Institute is supported in part by the Government of Canada through the Department of Innovation, Science and Economic Development and by the Province of Ontario through the Ministry of Colleges and Universities.
\end{acknowledgments}

\bibliography{ref.bib}

\newcommand{\modone}{~(\text{mod} 1)}
\begin{appendix}

\section{Quantum Error-Correcting Code Space}\label{app-qecc}

Following Ref.~\cite{kitaev2002classical}, one way to formulate the defining property for a quantum code (by definition, a subspace $\mathcal{C}\subseteq \mathcal{H}$) to correct arbitrary errors in a region $K$ is that, for any operators $\mathcal{O}_1$ and $\mathcal{O}_2$ acting on $K$ and any two orthogonal states $\ket{\xi_1}, \ket{\xi_2}\in \mathcal{C}$, we have:
\begin{equation}\label{eq-defQECC}
    \mathcal O_1\ket{\xi_1}\perp \mathcal O_2\ket{\xi_2}.
\end{equation}
Under this condition, two distinct (orthogonal) states are still distinct (orthogonal) after errors, so that in principle the pre-error states can be identified and reconstructed from the post-error states, and hence the name quantum error correcting code (QECC).
Indeed, it is well known that the above definition is equivalent to \emph{either} of the following properties \cite{PhysRevA.55.900}, as mentioned in the main text:
\begin{itemize}
    \item (quantum recoverability) there exists a quantum channel $\mathcal{R}$ such that $\mathcal{R}(\Tr_K\ket{\xi}\bra{\xi})=\ket{\xi}\bra{\xi}$ for $\forall\ket{\xi}\in\mathcal C$;
    \item (quantum indistinguishability) $\Tr_{K^c}\ket{\xi}\bra{\xi}$ is independent of $\ket{\xi}\in\mathcal C$.
\end{itemize}
Colloquially speaking, the erasure of a certain region is correctable if and only if the region contains no logical quantum information.

Now we prove the criteria \refeq{eq-QECcriteria} used in the main text -- namely, a subspace $\mathcal{C}$ spanned by vectors $\ket{\psi_i}$ (possibly unnormalized, non-orthogonal, or over-complete) is a QECC capable of correcting erasure of $K$ if and only if
\begin{equation}\label{eq-QECcriteria-app}
\Tr_{K^c}\ket{\psi_i}\bra{\psi_j}=\braket{\psi_j|\psi_i}\rho_K,~~\forall i,j.
\end{equation}
\begin{proof}
    It is clear that \refeq{eq-defQECC} is equivalent to:
\begin{equation}
    \braket{\xi_1|\mathcal{O}|\xi_2}=0 
\end{equation}
for $\forall \ket{\xi_1}\perp\ket{\xi_2}$ and $\forall\mathcal{O}$ acting on $K$.
Equivalently,
\begin{equation}\label{eq-appQEC}
    \Tr_{K^c}\ket{\xi_2}\bra{\xi_1}=0,~~\forall \ket{\xi_1}\perp\ket{\xi_2}.
\end{equation}

To show \refeq{eq-appQEC} from \refeq{eq-QECcriteria-app}, we just expand $\ket{\xi_1}$ and $\ket{\xi_2}$ using the basis $\{\ket{\psi_i}\}$.
For the opposite direction, picking an orthonormal basis $\{\ket{\xi_m}\}$ for $\mathcal{C}$ and applying  \refeq{eq-appQEC} to $\ket{\xi_m}\pm\ket{\xi_n}$, we find that $\Tr_{K^c}\ket{\xi_m}\bra{\xi_m}$ must be $m$-independent, denoted by $\rho_K$.
\refeq{eq-QECcriteria-app} is then proved from \refeq{eq-appQEC} and the above fact by decomposing $\ket{\psi_i}$ and $\ket{\psi_j}$ in the $\{\ket{\xi_m}\}$ basis.
\end{proof}

\section{Local Recoverability}
In this section, we show that Penrose tilings and Ammann-Beenker tilings can be uniquely recovered from their Ammann lines, and nonsingular Fibonacci quasicrystals can be uniquely recovered from the complement of any finite region.

\subsection{Penrose tilings}\label{app-PTA}
PTs can be uniquely reconstructed from their Ammann lines.
The reconstruction is based on the following observations:
\begin{itemize}
    \item 
    Two Ammann lines intersect at angle $\pi/5$ if and only if the intersection is the midpoint of a single-arrow edge. 
    This is because Ammann bars in the interior of a tile never intersect at $\pi/5$.
    \item 
    Two single-arrow edges meet at angle $4\pi/5$ if and only if they belong to the same thin rhombus.    
\end{itemize}

Hence, we can first bisect all $\pi/5$ angles and recover all single-arrow edges (their directions can also be recovered via Ammann lines).
Then we find all points where two single-arrow edges meet at $4\pi/5$ and recover all thin rhombi.
Thick rhombi are then automatically recovered.

\subsection{Ammann-Beenker tilings}\label{app-ABA}
Similarly, AB tilings can also be reconstructed from their Ammann lines.
Indeed, in an AB tiling,  rhombi are in one-to-one correspondence with intersections of Ammann lines at angle $\pi/4$ such that there are no more intersections within a distance of $1/2$ (in units of the side length of the rhombi).
To reconstruct, one first finds all such intersections and recovers all rhombi.
All squares then automatically appear, and their orientations may also be determined from the Ammann lines.

The periodic version of the AB tiling in \refsec{sec:AB} can be reconstructed in a similar fashion.

\subsection{Fibonacci Quasicrystals}\label{sec-recover-1D-infinite}
A 1D Fibonacci quasicrystal, represented as a bit string $\{a_n\}$, is generated (given a $\gamma\in \mathbb{R}\backslash\mathbb{Z}$) by calculating the decimal part of $\tau n+\gamma$ (here $\tau=\frac{\sqrt{5}-1}{2}$) as follows \cite{796d525f-e8d5-3468-a195-8c41cbb1f728}:
\begin{equation}\label{eq-appfib1}
    a_n=\begin{cases}
        1, \text{~if~}\tau n+\gamma \in[1-\tau,1) \modone \\
        0, \text{~if~}\tau n+\gamma \in[0,1-\tau) \modone,
    \end{cases}
\end{equation}
or
\begin{equation}\label{eq-appfib2}
    a_n=\begin{cases}
        1, \text{~if~}\tau n+\gamma \in(1-\tau,1] \modone \\
        0, \text{~if~}\tau n+\gamma \in(0,1-\tau] \modone.
    \end{cases}
\end{equation}
The two cases of $\gamma\equiv\tau n \modone ~(n\in\mathbb{Z})$ are singular, which we exclude.
For the nonsingular case, \refeq{eq-appfib1} and \refeq{eq-appfib2} coincide.

\begin{lemma}
If the set $\{\tau n_i \modone\}$ is dense in the unit circle $\mathbb{R}/\mathbb{Z}$ for a subset of integers $\{n_i\}$, then $\{a_{n_i}\}$ determines $\gamma\modone$.
\end{lemma}
\begin{proof}
    Suppose $\gamma\neq\gamma'$. 
    By the density of $\{\tau n_i \modone\}$, we can find an $n_i$ such that $1-\tau\in (\tau n_i+\gamma, \tau n_i+\gamma')$ (also mod 1).
    This implies $a_{n_i}\neq a'_{n_i}$, a contradiction.
\end{proof}

Now, for any $\{n_i\}$ that is the complement of a finite subset, $\{\tau n_i \modone\}$ is always dense. Hence any finite region (finite substring) of the 1D Fibonacci quasicrystal can be uniquely determined from its complement.

There is another, more algorithmic proof. 
The basic idea is that, if we are given a 1D Fibonacci quasicrystal except for a finite ``hole," we can use the deflation process to ``repair" the hole: we perform the deflation operation on as much of the tiling as we can ({\it i.e.}\ everywhere except in the immediate vicinity of the hole); and then iterate this deflation process.
After a finite number of deflations, the ``hole" has become comparable to (or smaller than) the size of the new supertiles produced by the deflation, and may be repaired (uniquely recovered) from knowledge of its neighborhood. 
For more details, see appendix \ref{app-recover-finite} (while the proof there is for a finite version, it can be modified straightforwardly for the infinite case).

\section{Local Indistinguishability}\label{app-indis}

Local indistinguishability is most easily understood from inflation and deflation. Given a finite patch $K$ of a Penrose tiling $T$, we can deflate $T$ enough times (thereby grouping the original tiles into larger and larger ``supertiles") so that $K$ is entirely contained in a single vertex configuration. Here, a vertex configuration means a vertex point as well as the elementary (supertile) rhombi that touch that vertex. There are only 7 possible vertex configurations (up to rotation) \cite{grunbaum1987tilings}. 
On the other hand, we can check that each vertex configuration must appear inside a
single supertile, as long as we inflate the tile sufficiently many times; and any Penrose tiling $T'$ must contain arbitrarily large supertiles, hence must contain all types of vertex configurations, hence must contain $K$.

The above argument can be made more quantitative to prove the strong local indistinguishability, namely, the {\it frequency} of a given patch $K$ is also independent of the Penrose tiling $T$. 
In appendix \ref{app-ABfinite} we show this for the periodic version of the Ammann-Beenker tiling. A similar proof also works for the Penrose tiling.

Moreover, the relative frequencies of the various local patterns are actually computable from the inflation rule alone.
In the following, as the simplest example, we show how to compute the frequencies of substrings in the discrete realization of the Fibonacci quasicrystal.

We rewrite the inflation rule $(L, S)\to(LS, L)$ as a substitution matrix:
\begin{equation}
    M^{(1)}=\begin{pmatrix}
    1~1\\
    1~0
    \end{pmatrix}.
\end{equation}
The first and second columns represent the inflation $L\to LS$ and $S\to L$ respectively.
If we start with a string with $x_0$ $``L"$s and $y_0$ $``S"$s, then after $n$ steps of inflation, the number of $``L"$s and $``S"$s are given by $(x_n, y_n)^T=M_1^n(x_0, y_0)^T$. 
In the $n\to\infty$ limit, the relative frequencies are determined by the eigenvector of $M^{(1)}$ corresponding to the unique largest eigenvalue (the ``Perron-Frobenius eigenvector"):
\begin{equation}\label{eq-frequency1}
\nu_1(L)=\frac{\sqrt{5}-1}{2}~~,\nu_1(S)=\frac{3-\sqrt{5}}{2}.
\end{equation}

To calculate the relative frequencies of length-2 substrings, we construct the \textit{induced substitution matrix} as follows. 
For each legal length-2 string $\omega\in\{LL, LS, SL\}$ (no $SS$), we construct the inflation $\omega'$ of $\omega$ (so $\omega'$ is $LSLS$, $LSL$ or $LLS$, respectively), and then we list the first $k$ (overlapping) length-2 strings in $\omega'$, where $k$ is the number of symbols in the inflation of the {\it first} digit of $\omega$ (so $k$ is $2$, $2$ and $1$, respectively).
This gives us the induced inflation rule: 
\begin{equation}
\begin{aligned}
    &(LL)\to \{(LS),(SL)\}\\
    &(LS)\to \{(LS),(SL)\} \\
    &(SL)\to \{(LL)\},
\end{aligned}
\end{equation}
and hence the induced inflation matrix
\begin{equation}
    M^{(2)}=\begin{pmatrix}
    0~0~1\\
    1~1~0\\
    1~1~0
    \end{pmatrix}.
\end{equation}
Its Perron-Frobenius eigenvector determines the relative frequencies of the length-2 substrings:
\begin{equation}\label{eq-nu2}
    \nu_2(LL)=\sqrt{5}-2,~\nu_2(LS)=\nu_2(SL)=\frac{3-\sqrt{5}}{2}.
\end{equation}

Another way to determine the above frequencies is to use the ``parent strings" in \refsec{sec-finiteindis}.
For example, since (1) the number of $LL$s in a quasicrystal must equal the number $S$s in the deflated quasicrystal and (2) the total number of letters increases by a factor of $\frac{\sqrt{5}+1}{2}$ under inflation, it follows that $\nu_2(LL)=\nu_1(S)/\frac{\sqrt{5}+1}{2}=\sqrt{5}-2$, consistent with \refeq{eq-nu2}.

The relative frequencies of longer substrings may be determined similarly.

\section{Periodic Ammann-Beenker-like Tilings}\label{app-ABfinite}

\newcommand{\ABedgeA}{\vcenter{\hbox{\includegraphics[height=\baselineskip]{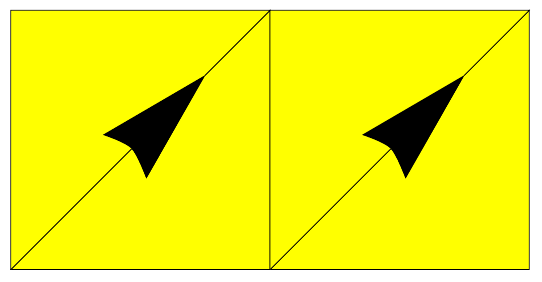}}}}
\newcommand{\ABedgeB}{\vcenter{\hbox{\includegraphics[height=\baselineskip]{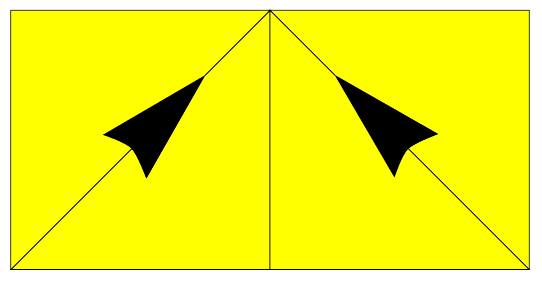}}}}
\newcommand{\ABedgeC}{\vcenter{\hbox{\includegraphics[height=2\baselineskip]{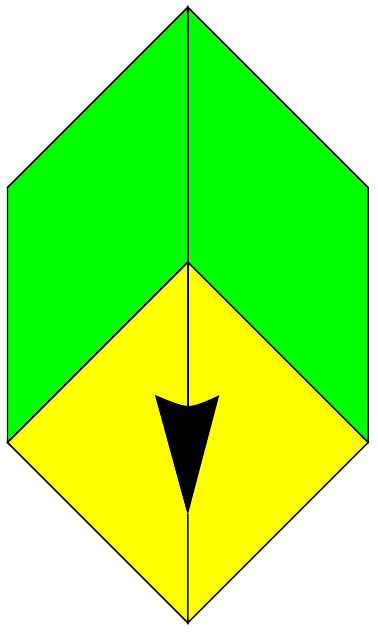}}}}
\newcommand{\ABedgeD}{\vcenter{\hbox{\includegraphics[height=2\baselineskip]{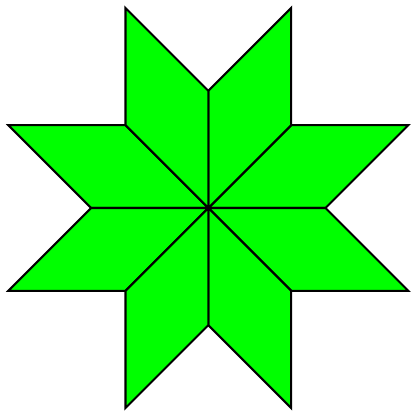}}}}

In this section, we prove the local indistinguishability of the periodic AB-like tilings constructed in \refsec{sec:AB}.

\begin{proposition}\label{claim-AB}
    For $n\geq 2$, there exists $r_n=\Theta((1+\sqrt{2})^n)$ such that for any disk $K$ of radius $r_n$, the pattern of $A_i^{(n)}$ inside $K$ is legal,
    and the number of appearances of this disk pattern in $A_i^{(n)}$ is independent of $i$.
\end{proposition}
Here, for ease of writing, we say a configuration is ``legal" if it also appears in a genuine (infinite, aperiodic) Ammann–Beenker tiling.
We will prove the proposition by reducing it to the case where $K$ is a ``vertex configuration".
Here, a vertex configuration is a vertex point of $A_i^{(2)}$ (called the ``center" of the vertex configuration) together with the tiles (squares and rhombi) touching it. For example, $\ABedgeC$ and $\ABedgeD$ are vertex configurations. The number of possible vertex configurations is finite.

\begin{proof}   
    For any vertex configuration, denoted by $V^{(2)}$, we denote its $(n-2)^{\text{th}}$ inflation as $V^{(n)}$ and call it a supervertex configuration (or just supervertex, for short); it is composed of a vertex point together with the supersquares and/or superrhombi surrounding it.
    We choose $r_n$ so that any disk $K$ of radius $r_n$ must be entirely contained in at least one supervertex. 
    According to \reffig{fig:ABinflationrule}, after each inflation, the linear size grows by a factor of $1+\sqrt{2}$, the square root of the largest eigenvalue of the inflation matrix $M=\begin{pmatrix}
    3~2\\
    4~3
    \end{pmatrix}$ obtained from \reffig{fig:ABinflationrule}.
    Hence we can choose $r_n=(1+\sqrt{2})^{n-2}r_2$.
    
    Now let us consider how a disk configuration $K$ could appear in $A^{(n)}$.
    We view $A^{(n)}$ as a union (with overlap) of supervertices. 
    If a supervertex contains $K$, then we call such supervertex an extension of $K$ and we say $K$ appears in $A^{(n)}$ via such extension. 
    (There could be more than one extension of a given copy of $K$;
    moreover, if a given supervertex contains more than one copy of $K$ at different locations, these should be counted as different extensions.)
    
    To count the number of $K$ in $A^{(n)}$, we only need to list all possible extensions of $K$ and count the number of each extension in $A^{(n)}$.
    Each extension $V^{(n)}$ belongs to one of the following four types:
    (1) $K$ is entirely contained in a supersquare in $V^{(n)}$;
    (2) $K$ is entirely contained in a superrhombus in $V^{(n)}$;
    (3) $K$ is entirely contained in two neighboring supertiles in $V^{(n)}$ and $K$ crosses the boundary between them;
    (4) others (where $K$ overlaps the ``central" vertex of the supervertex configuration).
    If we denote the number of type-$j$ supervertices in $A^{(n)}$ as $m_j~(j=1,2,3,4)$, the number of $K$ in $A_i^{(n)}$ is given by:
    \begin{equation}
        \frac{1}{4}m_1+\frac{1}{4}m_2+\frac{1}{2}m_3+m_4.
    \end{equation}
    The prefactor $\frac{1}{4}$ before $m_1$ is because, if $K$ is contained in $A^{(n)}$ via a type-1 supervertex, it is also contained in three other supervertices centered on the other three corners of the supersquare. The prefactor $\frac{1}{2}$ before $m_3$ is because the ``superedge" between the two supertiles crossed by $K$ in a type-3 supervertex is shared by two supervertices.

    Therefore, to prove Prop. \ref{claim-AB}, we can assume that $K$ is a supervertex.
    And since the supervertices are inflations of vertex configurations in $A^{(2)}$, we only need to prove the claims for $n=2$.
\end{proof}
\begin{lemma}
    Each vertex configuration in $A_i^{(2)}~(i=1,2,3,4)$ is legal, and the number of times it appears in $A_i^{(2)}$ is independent of $i$.
\end{lemma}
\begin{proof}
    Recalling that $A_i^{(2)}$ are the inflations of $A_i^{(0)}$, we can decompose each $A_i^{(2)}$ into 4-by-4 supersquares.
    Accordingly, there are three types of vertex configurations $V^{(2)}$:
    (1) $V^{(2)}$ is entirely contained in a supersquare;
    (2) the center of $V^{(2)}$ is on the boundary of two supersquares;
    (3) the center of $V^{(2)}$ is a corner of a supersquare.

    By construction, for any square orientation, the number of such squares in $A_i^{(0)}$ is independent of $i$, hence $A_i^{(2)}$ for $i=1,2,3,4$ contain the same number of type-1 vertices.
    Moreover, a square is legal, hence any type-1 vertex is also legal.

\newcommand{\ABedgeha}{\vcenter{\hbox{\includegraphics[height=0.35cm]{fig/ABboundary1.pdf}}}}
\newcommand{\ABedgehb}{\vcenter{\hbox{\includegraphics[height=0.35cm]{fig/ABboundary2.pdf}}}}
\newcommand{\ABedgehc}{\vcenter{\hbox{\includegraphics[height=0.35cm]{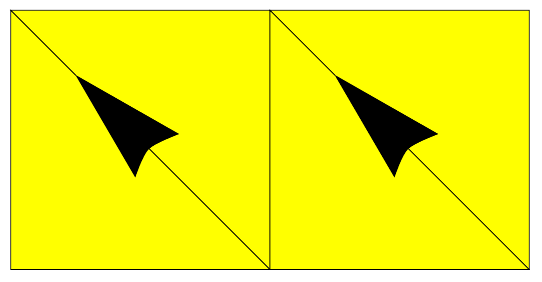}}}}
\newcommand{\ABedgehd}{\vcenter{\hbox{\includegraphics[height=0.35cm]{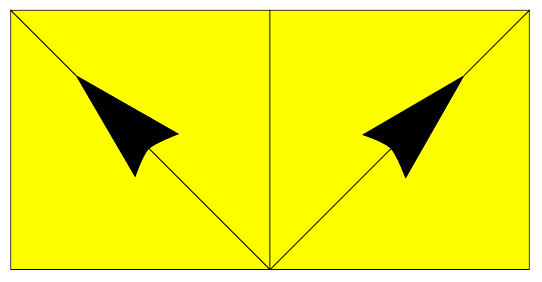}}}}
\newcommand{\ABedgehe}{\vcenter{\hbox{\includegraphics[height=0.35cm]{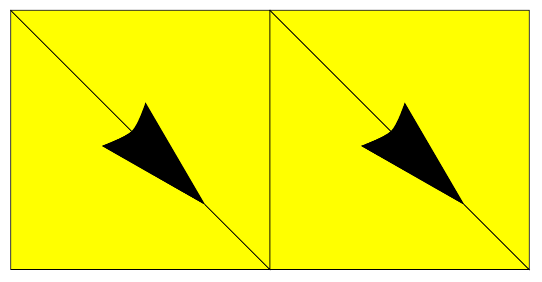}}}}
\newcommand{\ABedgehf}{\vcenter{\hbox{\includegraphics[height=0.35cm]{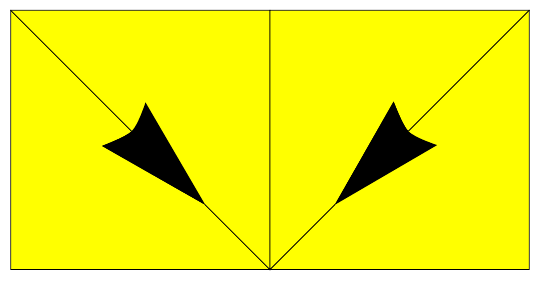}}}}
\newcommand{\ABedgehg}{\vcenter{\hbox{\includegraphics[height=0.35cm]{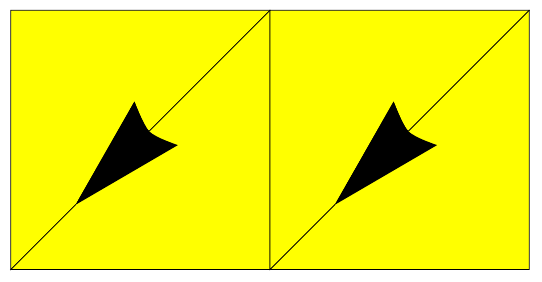}}}}
\newcommand{\ABedgehh}{\vcenter{\hbox{\includegraphics[height=0.35cm]{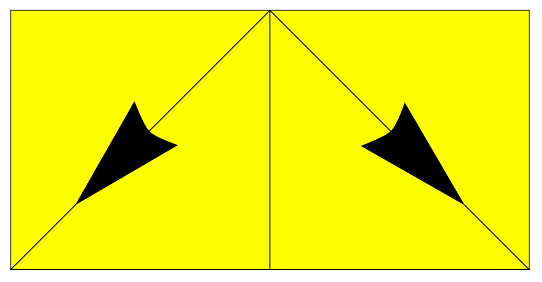}}}}
\newcommand{\ABedgeva}{\vcenter{\hbox{\includegraphics[height=0.7cm]{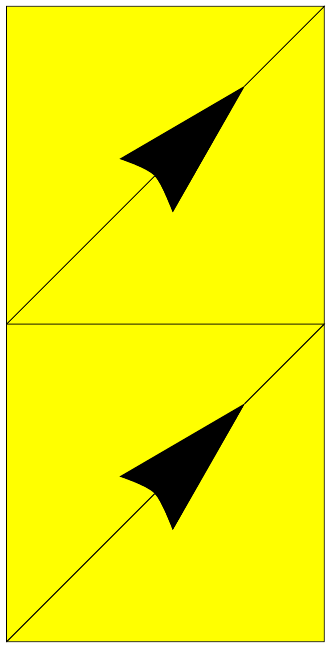}}}}
\newcommand{\ABedgevb}{\vcenter{\hbox{\includegraphics[height=0.7cm]{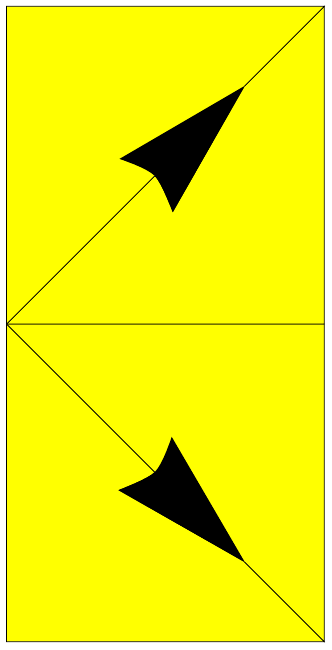}}}}
\newcommand{\ABedgevc}{\vcenter{\hbox{\includegraphics[height=0.7cm]{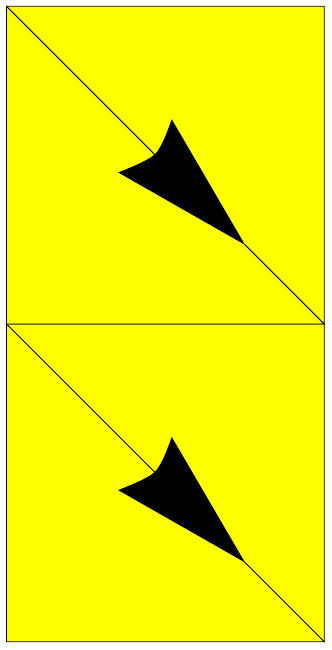}}}}
\newcommand{\ABedgevd}{\vcenter{\hbox{\includegraphics[height=0.7cm]{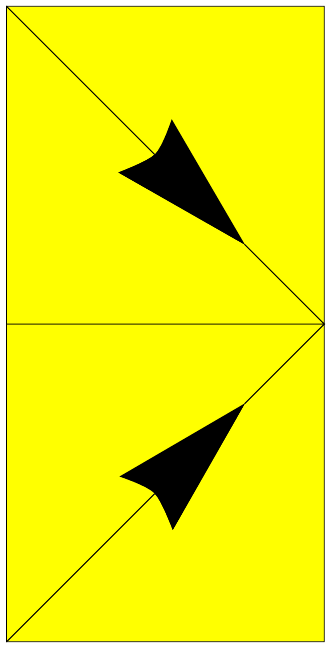}}}}
\newcommand{\ABedgeve}{\vcenter{\hbox{\includegraphics[height=0.7cm]{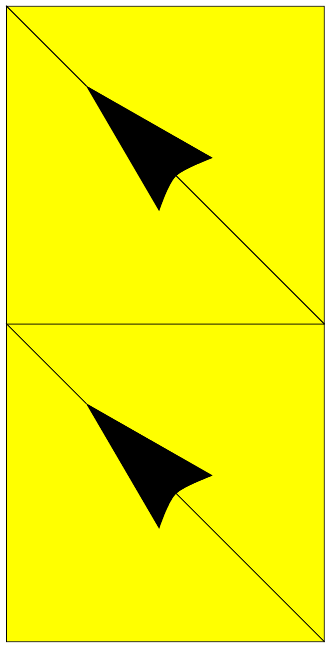}}}}
\newcommand{\ABedgevf}{\vcenter{\hbox{\includegraphics[height=0.7cm]{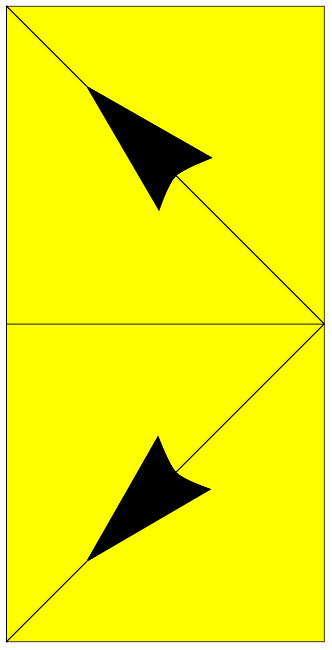}}}}
\newcommand{\ABedgevg}{\vcenter{\hbox{\includegraphics[height=0.7cm]{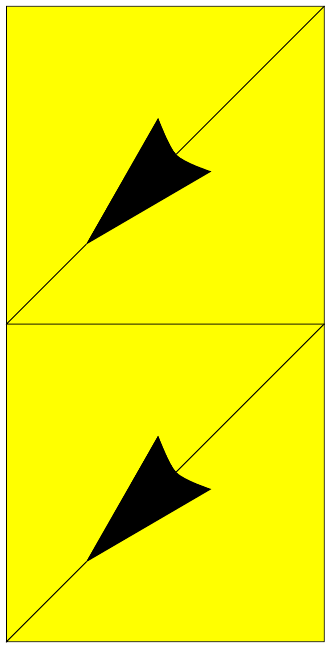}}}}
\newcommand{\ABedgevh}{\vcenter{\hbox{\includegraphics[height=0.7cm]{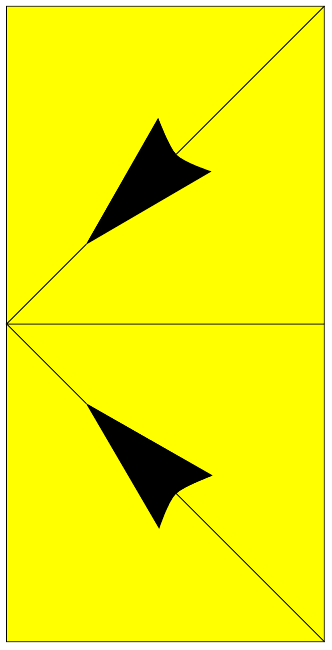}}}}
\newcommand{\ABedgeIa}{\vcenter{\hbox{\includegraphics[height=0.7cm]{fig/ABboundaryD.pdf}}}}
\newcommand{\ABedgeIb}{\vcenter{\hbox{\scalebox{1}[-1]{\includegraphics[height=0.7cm]{fig/ABboundaryD.pdf}}}}}
\newcommand{\ABedgeIc}{\vcenter{\hbox{\rotatebox{-90}{\includegraphics[height=0.7cm]{fig/ABboundaryD.pdf}}}}}
\newcommand{\ABedgeId}{\vcenter{\hbox{\rotatebox{90}{\includegraphics[height=0.7cm]{fig/ABboundaryD.pdf}}}}}
    \renewcommand\cellset{\renewcommand\arraystretch{1.1}\setlength\extrarowheight{1pt}}

\begin{table}
    \centering
    \begin{tabular}{c|c|c|c|c|c|c|c|c}
    \hline\hline
         & \makecell{$\ABedgeha$\\$\ABedgehc$} & \makecell{$\ABedgehb$\\$\ABedgehd$}  & \makecell{$\ABedgehe$\\$\ABedgehg$} & \makecell{$\ABedgehf$\\$\ABedgehh$} & $\ABedgeva$ $\ABedgevc$ & $\ABedgevb$ $\ABedgevd$ & $\ABedgeve$ $\ABedgevg$ & $\ABedgevf$ $\ABedgevh$\\
         \hline
         & \multicolumn{2}{c|}{$\ABedgeIa$} & \multicolumn{2}{c|}{$\ABedgeIb$} &\multicolumn{2}{c|}{$\ABedgeIc$}& \multicolumn{2}{c}{$\ABedgeId$}\\
         \hline
         $A_1^{(0)}$& 2 & 2 & 2 & 2 & 2 & 2 & 2 & 2\\
         \hline
         $A_2^{(0)}$& 0 & 4 & 0 & 4& 0 & 4 & 0 & 4\\
         \hline
         $A_3^{(0)}$& 2 & 2 & 2 & 2 & 0  & 4 & 0 & 4 \\
         \hline
         $A_4^{(0)}$& 0 & 4 & 0 & 4 & 2 & 2 & 2 & 2\\
         \hline\hline
    \end{tabular}
    \caption{Counting different edges in $A^{(0)}_i$ for 16 classes of edges, each specified by the orientation of two adjacent squares. For instance, $A_1^{(0)}$ contains two edges of class $\ABedgeA$ and two edges of class $\ABedgeB$, while $A_2^{(0)}$ contains zero of the former and four of the latter. After one inflation, both edges become $\ABedgeC$, hence the initial difference is inconsequential.}
    \label{tab:ABproof}
\end{table}

    For type-2, let us consider the edges of squares in $A_i^{(0)}$ (each contains 32 edges). 
    These edges belong to different classes based on the orientations of the (two) adjacent squares.
    The key point is that, although the number of edges of a specific type in $A_i^{(0)}$ could depend on $i$, such differences vanish once we go to $A_i^{(1)}$.
    In \reftab{tab:ABproof}, we list the number of different edges in each $A^{0}_i$. We see that, by the time we get to $A_i^{(1)}$, all type-2 vertices have already become legal, and the number of occurrences of each sort of type-2 vertex has already become $i$-independent.
    
    For type-3, we need to go to $A^{(2)}$. After two inflations, all corners become $\ABedgeD$, which is now legal. Hence the lemma also holds for type-3 vertices.    
\end{proof}

\section{Finite\&Discrete Realization}\label{app-finite}

In this section, we describe a QECC constructed on periodic qubit systems, using the Fibonacci symbolic substitutions as in \refsec{sec:discrete}.

\subsection{The construction}
We start from a cyclic bit string $F^{(0)}$ (strings related by translation are considered the same).
Applying the Fibonacci inflation rule $(1,0)\to(10,1)$ $n$ successive times,
we obtain a cyclic bit string $F^{(n)}$.
The length of $F^{(n)}$ is $|F^{(n)}|=k_0f_{n}+k_1f_{n+1}$, where $k_0$ and $k_1$ are the number of 0s and 1s in $F^{(0)}$, and $f_n$ is the $n^{\text{th}}$ Fibonacci number (defined by the recurrence $f_{n}=f_{n-1}+f_{n-2}$ and the initial values $f_0=f_1=1$).

Such finite strings share similar properties as the finite AB tilings in \refsec{sec:AB}. They locally resemble genuine infinite Fibonacci strings and are mutually locally indistinguishable:
\begin{itemize}
    \item any length $f_{n+1}$ substring of $F^{(n)}$ is also a substring of a genuine infinite Fibonacci string and vice versa;
    \item the number of appearances of any length $f_{n+1}$ substring $K$ in $F^{(n)}$ only depends on $k_0, k_1$ and $K$ (in other words, it is independent of the details of $F^{(0)}$ once $k_0$ and $k_1$ are fixed).
\end{itemize}
Moreover, there is also a remnant of local recoverability: 
given the knowledge of $k_0$ and $k_1$, 
an erased contiguous region $K$ of $F^{(n)}$ such that $|K|\leq f_{n}+1$ can be recovered from $F^{(n)}\backslash K$, 
up to a single 01 swapping in $K$, 
which corresponds to\footnote{Note that the $n$-step inflations of 01 and 10 are identical except the last two digits. This is the origin of the singular cases mentioned in \refsec{sec:discrete}} a single 01 swapping in $F^{(0)}$.

Now we can construct 
a QECC as follows. 
We fix $(k_0,k_1,n)$ and pick a set of cyclic strings $\mathcal{F}$, such that (1) each $F^{(0)}\in\mathcal{F}$ contains $k_0$ 0s and $k_1$ 1s;
(2) swapping any adjacent 1 and 0 in any string $F^{(0)}$, the new string is no longer in the set $\mathcal{F}$.
For any $F^{(0)}\in\mathcal{F}$, define the state $\ket{\Psi^{(n)}_{F^{(0)}}}$ as the superposition of all translations of $F^{(n)}$:
\begin{equation}\label{eq:finiteinflation}
    \ket{\Psi^{(n)}_{F^{(0)}}}\propto \sum_{x=1}^{|F^{(n)}|} \ket{x+F^{(n)}}.
\end{equation}

Local indistinguishability and local recoverability (proved below) again guarantee that the above states span the code space ${\cal C}$ of a QECC such that errors in any contiguous region $|K|\leq f_{n}+1$ can be corrected.  

\subsection{Higher Dimension Generalization}
The above construction can be generalized to higher dimensions via a ``Cartesian product" construction. Here, as an example, we describe the 2D case.

For two bit strings $(a_i)$ and $(b_j)$, we define their Cartesian product $a\times b$ as a table $P_{ij}$ of ordered pairs $P_{ij}=(a_{i},b_{j})$; or, equivalently, a table of quaternary bits $T_{ij}$, such that $T_{ij}=2a_i+b_j$.  The above Cartesian product is better understood if we transpose one of the bit strings. For example, 
\begin{equation}
    \begin{pmatrix}
        1\\0
    \end{pmatrix}\times (1,0,1)
    =\begin{pmatrix}
        3&2&3\\
        1&0&1
    \end{pmatrix}.
\end{equation}
Geometrically, such a construction amounts to taking the Cartesian product of two 1D point sets (each living on a circle) to construct a 2D point set (living on a torus). 

Now simply taking the Cartesian product of two wavefunctions in \refeq{eq:finiteinflation}, we get a qudit wavefunction in 2D, where for example
\begin{equation}
    \left\lvert\begin{matrix}
        1\\0
    \end{matrix}\right\rangle\times\ket{101}=\left\lvert\begin{matrix}
        3&2&3\\
        1&0&1
    \end{matrix}\right\rangle.
\end{equation}
Local indistinguishability and local recoverability for 1D bit strings guarantee that such 2D states span a quantum error-correcting code space such that errors in any squares of side length $f_{n}+1$ can be corrected.

\subsection{Proofs}
\subsubsection{Indistinguishability}\label{sec-finiteindis}
    \begin{proposition}
    For a string $K$ of length $f_{n+1}$, its number of occurrences in $F^{(n)}$ only depends on $K$, $k_0$, $k_1$. 
    \end{proposition}
    This claim is parallel to Prop.\ref{claim-AB} and a similar proof also works here. 
    Here, we give an alternative proof that gives a better constant, and also motivates the recovery algorithm in appendix \ref{app-recover-finite}.
    \begin{proof}
    We prove the claim by induction. 
    The case of $n=0$ is obvious: $f_{0+1}=1$, so we only consider single digits. The number of appearance of $K=0$ (or $K=1$) in $F^{(0)}$ is exactly $k_0$ (or $k_1$) by definition.

    For $n\geq 1$, we can relate the number of appearances of $K$ in $F^{(n)}$ to the number of appearance of another substring in $F^{(n-1)}$.
    More precisely, we define its ``parent string" $\mathcal{D}(K)$, 
    so that  for any $F^{(0)}$ the multiplicity of $K$ in $F^{(n)}$ equals the multiplicity of $\mathcal{D}(K)$ in $F^{(n-1)}$:
    \begin{equation}\label{eq-samemulti}
    \#K\text{~in~} F^{(n)}=\#\mathcal{D}(K) \text{~in~} F^{(n-1)}~~(\forall F^{(0)}).
    \end{equation}
    $\mathcal{D}(K)$ is solely determined by $K$.
    
    To define $\mathcal{D}(K)$, we just need to perform one deflation step. 
    Namely, we apply the following procedure to $K$:
    \begin{enumerate}
        \item insert a virtual cut $\urcorner\ulcorner$ to the left of each 1;
        \item inside the string: replace each $\ulcorner 10\urcorner$ by a $\ulcorner 1\urcorner$, replace each $\ulcorner 1\urcorner$ by a $\ulcorner 0\urcorner$;
        \item at the two endpoints: replace $0\urcorner$ by $1\urcorner$ if $K$ starts with 0 (because there must be a 1 to the left of any 0, which together deflates to 1); delete $\ulcorner 1$ at the end if $K$ ends with 1 (since the inflation of both 1 and 0 starts with at 1, so there is no constraint on this digit in $F^{(n-1)}$).
    \end{enumerate}
    The above procedure can always be applied for any string $K$ that could appear as a substring of a $n$-step ($n\geq 1$) inflation of some string.
    For example, for $K=010, 110, 011, 101$ (these are all possible 3 digit strings for $n\geq 2$), the corresponding $\mathcal{D}(K)=11, 01, 10, 1$. 
    It is evident by construction that the above defined $\mathcal{D}(K)$ satisfies \refeq{eq-samemulti}.

    From \refeq{eq-samemulti}, the desired property for string $K$ at level $n$ is equivalent to the same property for its parent string $\mathcal{D}(K)$ at level $(n-1)$.
    Therefore, the induction is concluded by the following lemma.
    \end{proof}
    \begin{lemma}
    For any string $K$ of length $f_{n+1}$ that could appear in $F^{(n)}$ (here $n\geq 1$), its parent string satisfies $|\mathcal{D}(K)|\leq f_{n}$.
    \end{lemma} 
 \begin{figure}
    \centering
    \includegraphics[width=\columnwidth]{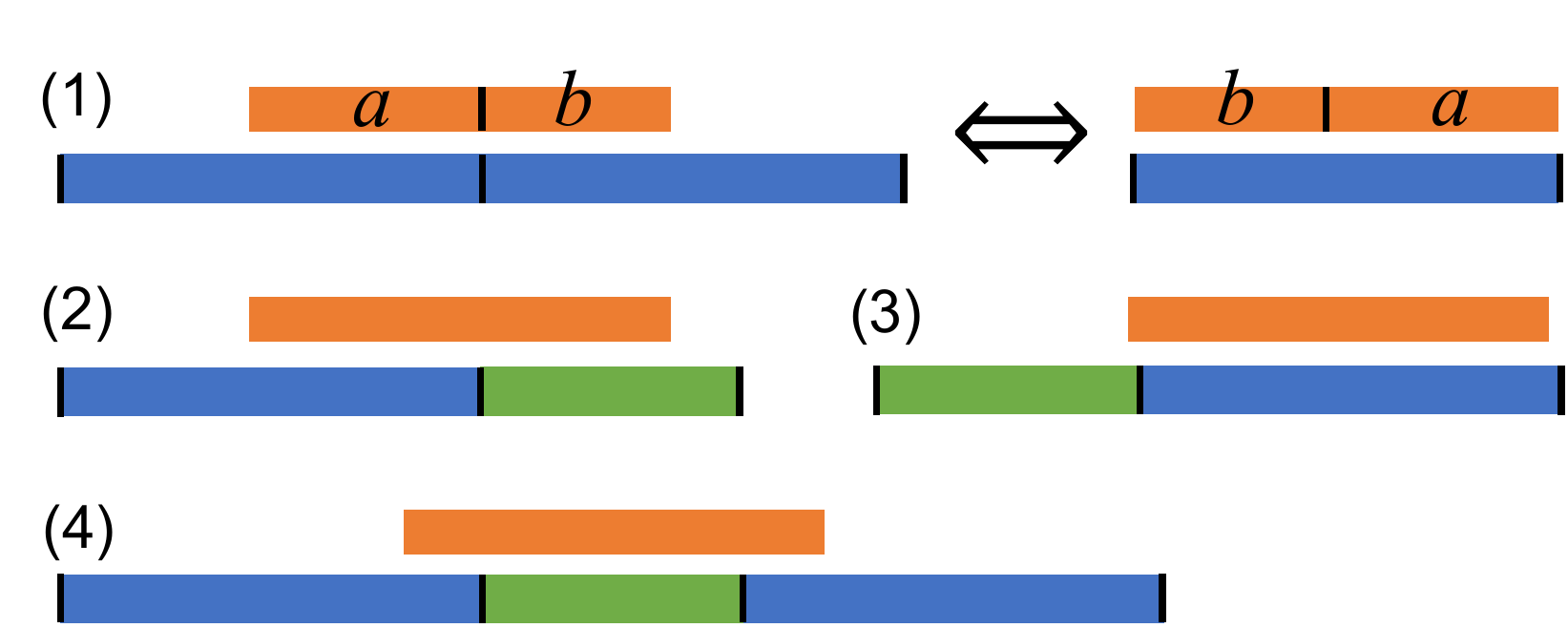}
    \caption{Proof of $|\mathcal{D}(K)|\leq f_n$ when $|K|=f_{n+1}$. The blue\allowbreak/green\allowbreak/orange intervals denote $Q_{n+1}$, $Q_n$, $K$ respectively. $|K|=|Q_{n+1}|=f_{n+1}$, $|Q_{n}|=f_{n}$. Depending on the position of $K$ in $F^{(n)}$, there are four cases.}
    \label{fig:move}
\end{figure}
\begin{proof}
    We denote the $n^{\text{th}}$ inflation of 0 as $Q_{n}$, which is also the $(n-1)^{\text{th}}$ inflation of 1.
    $|Q_n|=f_n$.
    By considering the length, we know $|K|=f_{n+1}$ implies that $K\subset Q_{n+1}Q_{n+1}$ or $Q_{n+1}Q_{n}$ or $Q_{n}Q_{n+1}$ or $Q_{n+1}Q_nQ_{n+1}$.
    Recall that to find $\mathcal{D}(K)$ from $K$, one just needs one step of deflation, which is a local procedure.
    \begin{enumerate}
        \item Case $Q_{n+1}Q_{n+1}$. Following the notation in \reffig{fig:move}, we decompose $K$ as $ab$.
        If we move $b$ to the left of $a$, then $ba=Q_{n+1}$.
        Deflation is locally well-defined, hence performing deflation on $ab$ is equivalent to that on $Q_{n+1}$\footnote{More precisely: 
        if $a$ starts with 1 and $b$ ends with 0, or if $a$ starts with 0(which implies that $b$ ends with 1),      
        then $\mathcal{D}(ab)$ and $Q_n$ are related by the same move; 
        if $a$ starts with 1 and $b$ ends with 1, 
        then $\mathcal{D}(ab)$ is 1 digit less that $Q_n$.
        Also note that $\mathcal{D}(Q_{n+1})$ does not always equals $Q_n$. For example, when $n=2$, $Q_{n+1}=101$, $Q_n=10$, but $\mathcal{D}(K)=1.$}.
        So $|\mathcal{D}(K)|\leq |Q_n|=f_n$.
        
        \item Case $Q_{n+1}Q_n$. 
        Note that $Q_n$ can be regarded as the initial $f_n$ digits of $Q_{n+1}$, so the problem is reduced to the case $Q_{n+1}Q_{n+1}$.
        
        \item Case $Q_nQ_{n+1}$.
        Note that $Q_nQ_{n+1}$ is exactly the same as $Q_{n+1}Q_n$ except the last two digits (where 0 and 1 are swapped).
        If $K$ does not include any of the two digits, or if $K$ includes both digits, then we are in the case of $Q_{n+1}Q_n$.
        If $K$ only includes one of the two digits, then the start and end of $K$ are the same. One can check that, regardless of whether it is 0 or 1, we always have $|\mathcal{D}(K)|=|Q_n|$.
        \item Case $Q_{n+1}Q_nQ_{n+1}$.
        Note that the first $(f_{n+1}-1)$ digits of $Q_nQ_{n+1}$ are the same as the first $(f_{n+1}-1)$ digits of $Q_{n+1}$, so the problem is reduced to the case of $Q_{n+1}Q_{n+1}$.
    \end{enumerate}
\end{proof}

\subsubsection{Recoverability}\label{app-recover-finite}

\begin{proposition}\label{prop:recoverfinite}
If $k_0$, $k_1$ and $n$ are known, a string $F^{(n)}$ can be recovered after erasing a contiguous region $K\subset F^{(n)}$ of length $|K|=f_{n}+1$,  
up to a single 01 swap in $K$, 
corresponding to a single 01 swap in $F^{(0)}$.
\end{proposition}   
    This claim is a finite version of the one in \refsec{sec-recover-1D-infinite}. 
    Here, we give an algorithmic proof, which complements the existence proof there.
    The main idea is to perform local deflation as much as possible, trying to recover digits in the ``initial layer" $F^{(0)}$. 
    \begin{lemma}[single digit recovery by deflation]\label{lemma-singledigit}
        To recover a digit in $F^{(0)}$ by deflation, it is enough to know $f_{n+2}-1$ digits starting from the left endpoint of its descendant substring (supercell) in $F^{(n)}$.
    \end{lemma}
    \begin{proof}[Proof of lemma \ref{lemma-singledigit}]
     By induction. The case of $n=1$ is obvious: $f_{n+2}-1=2$, and indeed knowing 2 digits starting from the left endpoint of its inflation is enough to determine the parent digit. 
     For larger $n$, by the above $n=1$ case, to determine a single digit $F^{(0)}$, it is enough if we have determined 2 digits in $F^{(1)}$ starting from the left endpoint of its inflation.
     A sufficient condition is to know the $f_n$ digits in $F^{(n)}$ corresponding to the descendant of the first digit (must be 1), as well as $f_{n+1}-1$ digits (by induction) to determine the second digit.
     Note that $f_{n}+f_{n+1}-1=f_{n+2}-1$, so we are done.
    \end{proof}

\begin{proof}[Proof of Prop.\ref{prop:recoverfinite}]
   
    Let $a_{i}$ denote the $i$th digit in $F^{(0)}$, so $F^{(0)}=[\cdots a_0a_1a_2\cdots]$, and let $P_i$ denote the descendant\footnote{Someone reconstructing $F^{(0)}$ does not know the locations of the $P_i$ {\it a priori}, but learns their locations via the recovery process.}
    of $a_i$ in $F^{(n)}$. 
    $|P_i|=f_n$ if $a_i=0$ and $|P_i|=f_{n+1}$ if $a_i=1$.  
    The assumption $|K|\leq f_n+1$ implies that $K$ at most intersects two pieces, say $K\subset P_1\cup P_2$. 
    Since $3f_n\geq f_{n+2}-1$, Lemma \ref{lemma-singledigit} implies that all digits can be determined by deflation except $a_{-1}, a_0, a_1, a_2$. 
    Moreover, the relative position of $K$ in $P_{-1}\cup P_0\cup P_{1}\cup P_2$ can also be determined. 
    
    If $K$ does not include the rightmost digit of $P_2$, then $a_2$ can be determined (simply by looking at the last digit of $P_2$), and we are left with $a_{-1}, a_0, a_1$ and $K\backslash P_2\subset P_1$. \
    If $K$ includes the rightmost digit of $P_2$, then since $f_{n+2}-1+f_{n}+1\leq 4f_n$ for $n\geq 1$, lemma \ref{lemma-singledigit} implies that $a_{-1}$ can also be determined by deflation.
    Therefore, in any case, we are always left with at most three unknown digits (we redefine them as $a_0, a_1, a_2$) and an unknown region $K\subset P_1P_2$, $|K|\leq f_n+1$, where $K$ includes the rightmost digit of $P_2$. 

    If $a_{0}=1$, then $|P_0|+|P_1|+|P_2|\geq f_{n+1}+2f_{n}=f_{n+2}-1+f_{n}+1$, hence lemma \ref{lemma-singledigit} implies that we can actually determine $a_{0}$ by deflation.
    Therefore, $a_{0}$ can be determined in any case: either directly by deflation, or if this simple method fails, it must be 0.

    Now we are left with only two unknown digits $a_2, a_3$ in $F^{(0)}$. 
    With the knowledge of $k_0$ and $k_1$, we can determine the set $\{a_2,a_3\}$.
    If $\{a_2,a_3\}=\{0\}$ or $\{1\}$, we are done.
    If $\{a_2,a_3\}=\{0,1\}$, then we have an ambiguity of a single 01 swap.
\end{proof}
   
\section{Entanglement Entropy}
\label{app-entropy}

It is standard in condensed matter physics and high energy physics to consider the entanglement entropy of many-body systems, defined as the von Neumann entropy of the reduced density matrix of subsystems.
In this Appendix, we comment on the entanglement entropy of the states in the code space of our QECCs, and their relation to the complexity function.

We take the discrete wavefunction \refeq{eq:1Dwavefunction} as an example.  (Calculating the entanglement entropy for systems with continuum degrees of freedom generally requires a cut-off, or regularization.)
For a subregion $K$ of length $n$, the reduced state $\rho_{K}$ is a classical mixture of possible substrings, where the coefficients are frequencies of appearance (similar to \refeq{eq:TT}). 
Define the \emph{complexity function} $p(n)$ as the number of possible substrings of length $n$, then the entropy of $\rho_{K}$, which equals the entanglement entropy between $K$ and $K^c$, is bounded by:
\begin{equation}
    S(n)\leq \log p(n).
\end{equation}
It turns out that $p(n)=n+1$ for 1D Fibonacci quasicrystals \cite{baake2013aperiodic}, hence
\begin{equation}
    S(n)=O(\log n).
\end{equation}

Actually, the substring frequencies of 1D Fibonacci quasicrystals can be calculated exactly.
Based on an observation in \refcite{lothaire2002algebraic}, we can inductively (full proof omitted) prove that 
\begin{proposition}
    If $n\in[f_{k-1},f_k-1]$, then among the $(n+1)$ possible substrings, $(n-f_{k-1}+1)$ of them have frequency $\tau^{k}$, $(n-f_{k-2}+1)$ have frequency $\tau^{k-1}$, and $(f_k-n-1)$ have frequency $\tau^{k-2}$ (here $\tau=\frac{\sqrt{5}-1}{2}$).
\end{proposition}

Therefore,
\begin{equation}
\begin{aligned}
    S(n)=&\Big[(n-f_{k-1}+1)k\tau^{k}+(n-f_{k-2}+1)(k-1)\tau^{k-1}\\
    &+(f_{k}-n-1)(k-2)\tau^{k-2}\Big]\log(\frac{1}{\tau}).
\end{aligned}
\end{equation}
Hence indeed
\begin{equation}
    S(n)=\log(n)+\Theta(1).
\end{equation}

\end{appendix}

\end{document}